\documentclass[prx,twocolumn,superscriptaddress]{revtex4-1}
\pdfoutput=1
\usepackage[utf8]{inputenc}
\usepackage{amsmath, amssymb,dsfont}
\usepackage{verbatim}
\usepackage[usenames,dvipsnames]{xcolor}
\usepackage{graphicx}
\usepackage[pdftex, plainpages=false]{hyperref}
\usepackage{mathtools}
\usepackage{stmaryrd, ifthen}
\usepackage{array}
\usepackage{tabularx,float}
\usepackage[export]{adjustbox}
\usepackage{bm,wasysym}
\usepackage{enumitem}
\usepackage{calc}
\usepackage{etoolbox}
\mathtoolsset{showonlyrefs=false}
\usepackage{footnote}
\usepackage{dcolumn}
\usepackage{microtype}
\usepackage{booktabs}
\usepackage{epstopdf}

\usepackage{euscript}

\definecolor{nblue}{rgb}{0.2,0.2,0.7}
\definecolor{ngreen}{rgb}{0.2,0.6,0.2}
\definecolor{nred}{rgb}{0.7,0.2,0.2}
\definecolor{nblack}{rgb}{0,0,0}

\usepackage{tikz}
\usetikzlibrary{positioning,intersections}
\usetikzlibrary{calc}
\usetikzlibrary{shapes.geometric}
\usepackage{braids}
\usepackage{tikz-cd}

\makeatletter
\newsavebox{\@brx}
\newcommand{\llangle}[1][]{\savebox{\@brx}{\(\m@th{#1\langle}\)}%
  \mathopen{\copy\@brx\kern-0.5\wd\@brx\usebox{\@brx}}}
\newcommand{\rrangle}[1][]{\savebox{\@brx}{\(\m@th{#1\rangle}\)}%
  \mathclose{\copy\@brx\kern-0.5\wd\@brx\usebox{\@brx}}}
\makeatother

\hypersetup{
    colorlinks,
    linkcolor={gray!80!black},
    citecolor={blue!90!black},
    urlcolor={gray!60!black}
}

\makeatletter

\usepackage{amsthm}

\theoremstyle{remark}
\newtheorem{definition}{Definition}
\newtheorem{example}[definition]{Example}

\newtheorem{remark}[definition]{Remark}
\newtheorem{theorem}[definition]{Theorem}
\newtheorem{corollary}[definition]{Corollary}
\newtheorem{principle}{Principle}

\newtheorem{observation}[definition]{Observationn}

\newcommand\Tr   {\operatorname{Tr}}
\newcommand{\one}{\mathds{I}}

\newcommand{\Acal}{\mathcal{A}}
\newcommand{\Bcal}{\mathcal{B}}

\newcommand{\Xcal}{\mathcal{X}}
\newcommand{\Ycal}{\mathcal{Y}}
\newcommand{\PDO}{\mathbf{PDO}}
\newcommand{\QPC}{\mathbf{QPC}}
\newcommand\Herm{\mathbf{Herm}}
\newcommand{\HPTP}{\mathbf{HPTP}}

\newcommand{\orcidd}[1]{\href{https://orcid.org/#1}{\includegraphics[width=8pt]{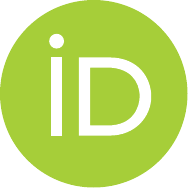}}}

\begin{document}

\title{Quantum space-time marginal problem: global causal structure from local causal information}

\author{Zhian Jia\orcidd{0000-0001-8588-173X}}
\email{giannjia@foxmail.com}
\affiliation{Centre for Quantum Technologies, National University of Singapore, Singapore 117543, Singapore}
\affiliation{Department of Physics, National University of Singapore, Singapore 117543, Singapore}

\author{Minjeong Song}
\email{song.at.qit@gmail.com}
\affiliation{Nanyang Quantum Hub, School of Physical and Mathematical Sciences, Nanyang Technological University, Singapore 637371, Singapore}

\author{Dagomir Kaszlikowski}
\email{phykd@nus.edu.sg}
\affiliation{Centre for Quantum Technologies, National University of Singapore, Singapore 117543, Singapore}
\affiliation{Department of Physics, National University of Singapore, Singapore 117543, Singapore}

\begin{abstract}
Spatial and temporal quantum correlations can be unified in the framework of the pseudo-density operators, and quantum causality between the involved events in an experiment is encoded in the corresponding pseudo-density operator. We study the relationship between local causal information and global causal structure.
A space-time marginal problem is proposed to infer global causal structures from given marginal causal structures where causal structures are represented by the reduced pseudo-density operators; we show that there almost always exists a solution in this case.
By imposing the corresponding constraints on this solution set, we could obtain the required solutions for special classes of marginal problems, like a positive semidefinite marginal problem, separable marginal problem, etc.
We introduce a space-time entropy and propose a method to determine the global causal structure based on the maximum entropy principle.
The notion of quantum pseudo-channel is also introduced and we demonstrate that the quantum pseudo-channel marginal problem can be solved by transforming it into a pseudo-density operator marginal problem via the channel-state duality.
\end{abstract}
\maketitle


\section{Introduction}

The relativity theory treats space and time on equal footing, and they are unified in the conception of the space-time manifold.
However, in the standard Copenhagen interpretation of quantum mechanics, space and time play extremely different roles.
This reflects in several differences between time and space: the time-energy uncertainty relation takes a different form from the position-momentum uncertainty relation \cite{Heisenberg1927};
we only have the probability distribution of particles over space and the time evolution of this distribution is controlled by Hamiltonian, there is no probability distribution over time \cite{bohm2012quantum};
the well-established formalism of tensor-product structure to represent states across space are not suitable for states in time 
\cite{Zanardi2004quantum,cotler2019locality,horsman2017can}, etc.
These differences need to be deeply understood especially when we are dealing with problems that both the relativity and quantum effects cannot be neglected like quantum black hole \cite{Harlow2016jerusalem} and relativistic quantum information \cite{Peres2004quantum,martin2011relativistic}.
Searching for a representation of quantum mechanics that treats space and time in a more even-handed fashion is thus a crucial problem and may shed new light on the notion of quantum space-time. 
There have been a variety of proposals for space-time states, 
process matrix \cite{oreshkov2012quantum}, consistent history \cite{griffiths1984consistent}, entangled histories \cite{cotler2016entangled}, and quantum-classical game \cite{gutoski2007toward}, superdensity operators \cite{cotler2018superdensity}, multi-time states \cite{Aharonov2009multi}, pseudo-density operator (PDO) \cite{fitzsimons2015quantum}, doubled density operator \cite{jia2023spatiotemporal}, etc. 
Among these proposals, PDOs turn out a convenient framework with broad applications in quantum information theory. They provide a direct generalization of the density operator and have been used in various areas, including quantum causal inference~\cite{liu2023quantum}, quantum communication~\cite{Pisarczyk2019causal}, temporal quantum teleportation~\cite{marletto2021temporal}, temporal quantum steering \cite{Uola2020quantum}, and more.

Clarifying the relation between the whole and its parts is crucial in many areas of science.
The question that considers in what situation the local information can be reproduced from a global structure 
is known as the  marginal problem.
The marginal problem has a long history. The \emph{probability distribution marginal problem} (or simply \emph{classical marginal problem}) considers the following question: given a family of sets of random variables $\{\Acal_1,\cdots,\Acal_n\}$ for which each $\Acal_i$ has their respective joint probability distribution $p_{\Acal_i}(X\in \Acal_i)$, and the marginals are compatible, \emph{viz.}, $\sum_{X\in \Acal_i \setminus (\Acal_i\cap \Acal_j)}p_{\Acal_i}=\sum_{Y\in \Acal_j\setminus (\Acal_i\cap \Acal_j)}p_{\Acal_j}$, if there exists a joint probability distribution $p_{\Acal}$ for all random variables $\Acal=\cup_i \Acal_i$ such that all $p_{\Acal_i}$ can be recovered as marginals of $p_{\Acal}$.
This seemingly effortless problem is indeed highly nontrivial, there exist locally compatible distributions that do not have global solutions. 
And the problem has been shown to be NP-hard \cite{pitowsky1989quantum}.
The classical marginal problem has broad applications in many fields, e.g., in quantum contextuality and Bell nonlocality \cite{bell1964, bell1966, kochen1967problem, Brunner2014bell, budroni2021quantum}.
It also has applications in the monogamy of quantum correlations \cite{Jia2016}, in statistical mechanics \cite{schlijper1988tiling}, and so on.
In quantum mechanics, states are represented by density operators, and thus the marginal problems are rephrased in terms of density operators.
The question of whether a given set of marginals (reduced density operators) is compatible with a global density operator is called a quantum state marginal problem, see, e.g \cite{schilling2015quantum}, and references therein.
This seemingly easy problem turned out to be challenging to solve in general, and it lies at the heart of many problems in quantum physics.

The classical and quantum marginal problems traditionally focus on spatially distributed events.
For example, in the classical Bell scenario, spatially distributed parties implement local measurements, and the measurement statistics must satisfy the non-signaling principle \cite{Brunner2014bell}.
When considering the temporal case, specifically in scenarios like the Leggett-Garg test \cite{leggett1985quantum}, the measurement statistics exhibit differences as only one-way non-signaling is applicable.
In the quantum case, the density operator marginal problem has been investigated from aspects, but the temporal quantum marginal problem has only been investigated within the context of the channel marginal problem \cite{haapasalo2021quantum, girard2021jordan, Hsieh2022quantum}.
In this work, we will investigate the quantum marginal problem in the general spatiotemporal setting using PDO formalism.
We will consider both the case of space-time states and higher-order dynamics.
We will show that in this case, there almost always exists a solution to the marginal problem, and the space-time correlations are polygamous in general.
Since PDO encodes the quantum causal structure for a given event set, solving the marginal problem can be regarded as the first step to inferring the global causal structure from the given local causal structures, this has many potential applications in quantum causal inference \cite{costa2016quantum,Giarmatzi2018a,bai2022quantum,cotler2019quantum,liu2023quantum}.
As an application of our results for the space-time marginal problem, we will explore how to infer the global PDO from the given local PDOs using an information-theoretic approach. We will introduce the entropy of spatiotemporal PDO and investigate how to obtain the best approximation of the global PDO using the maximum space-time entropy principle.

The rest of the paper is organized as follows. In Sec.~\ref{sec:causality}, we review the definition of PDO.
Sec.~\ref{sec:PDOMagProb} discusses the space-time PDO marginal problem.
We first show that there always exist a set of solutions in the space of Hermitian trace-one operators.
Then using this result, we discuss how to obtain the solution to the marginal problem by imposing corresponding constraints over the Hermitian trace-one solution set,  like positive semidefiniteness, separability, etc.
If the solution to a marginal problem is guaranteed, we could further ask: How much local information do we need to reconstruct global information?
This problem is investigated in Sec.~\ref{sec:InfGlob}.
By introducing the entropy of space-time states and the generalized maximum entropy principle, we briefly discuss how to infer  the global space-time state from the given set of reduced space-time states.
Finally, we conclude and outline some open problems and future directions.
The Appendices provide additional technical details and address the marginal problem of higher-order maps for PDOs.


\section{Pseudo-density operator}
\label{sec:PDOintro}

\begin{figure}[t]
	\centering
	\includegraphics[width=6.8cm]{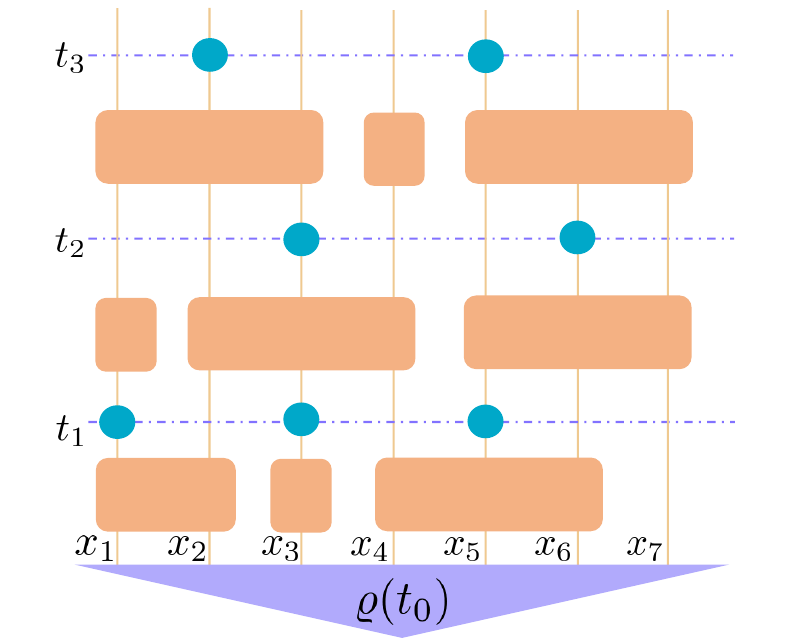}
	\caption{The depiction of the scenario of the pseudo-density operator. The vertical solid lines (quantum wires) represent local quantum freedoms, their labels can be regarded as the spatial coordinates. 
    The time instants are represented by horizontal dashed lines.
    Time flow is upwards.
    The purple triangle represents the input state at the initial time $t_0$.
    Between each two consecutive time instants, there are possibly some quantum operations implemented over the system, and the orange boxes represent the quantum gate given by the quantum channels. The light blue dots represent the space-time events $E(x,t)$, i.e., measuring (generalized) Pauli operators at some instants of time over some local quantum freedom.\label{fig:PDO}}
\end{figure}

In quantum mechanics, a density operator is usually regarded as a probabilistic mixture of pure quantum states. But it can also be viewed as a representation of the correlation functions of Pauli operators over the system, and the most famous one is the qubit Bloch vector representation \cite{Hioe1981N}.
For a multipartite system, each local Pauli operator is measured simultaneously. Thus the density operator only encodes the spatial correlations $T^{\mu_1,\cdots,\mu_n}=\langle \{\sigma_{\mu_1}(x_1), \cdots ,\sigma_{\mu_n}(x_n)\}\rangle=\Tr (\otimes_i\sigma_{\mu_i}\varrho)$:
\begin{equation}
    \varrho=\frac{1}{2^n} \sum_{\mu_i=0}^3 T^{\mu_1,\cdots,\mu_n} \sigma_{\mu_1}\otimes \cdots \otimes \sigma_{\mu_n}.
\end{equation}
It's natural to consider the situation where the local quantum degrees of freedom  are fixed and we measure them at different time instants $T^{\mu_1,\cdots,\mu_n}=\langle \{ \sigma_{\mu_1}(t_1), \cdots , \sigma_{\mu_n}(t_n)\}\rangle$.
This leads to the definition of PDO.
Thus a PDO generalizes the spatial correlation to admit causal structures with subsystems associated with the same degrees of freedom at different time instants \cite{fitzsimons2015quantum}.
Consider a single qubit state, a two-time PDO is characterized by the Pauli correlator $\langle \{\sigma_{\mu_1}(t_1),\sigma_{\mu_2}(t_2))\}\rangle$, this is obtained from the qubit state by implementing sequential measurements. E.g., to obtain $\langle \{\sigma_x(t_1),\sigma_x(t_2)\}\rangle$, we have 
\begin{equation}
\begin{aligned}
        p(x_1,x_2)=\Tr (\Pi_{x_2} \mathcal{E}(\Pi_{x_1}\varrho \Pi_{x_1}) \Pi_{x_2}),\\
       \langle \{\sigma_x(t_1),\sigma_x(t_2)\}\rangle=\sum_{x_1,x_2=\pm 1} x_1 x_2 p(x_1,x_2),
\end{aligned}
\end{equation}
where $\Pi_x$'s (with $x=\pm 1$) are the projector corresponding to Pauli X measurement, and $\mathcal{E}$ is the evolution channel between two time instants.
The general situation is depicted in Fig.~\ref{fig:PDO}, where the input state of the quantum circuit is a multipartite state $\varrho(t_0)$.
We choose several space-time local degrees of freedom (cyan dots in Fig.~\ref{fig:PDO}) to obtain Pauli correlators $T^{\mu_1,\cdots,\mu_n}=\langle \{\sigma_{\mu_1}(t_1),\cdots,\sigma_{\mu_k}(t_1), \sigma_{\mu_{k+1}}(t_2),\cdots \sigma_{\mu_n}(t+m)\}\rangle$.
We will call each point we choose to implement Pauli measurement a (space-time) event, and the set of all events is called the event set.
Based on this generalization, we obtain the PDO 
\begin{equation}
    R=\frac{1}{2^n} \sum_{\mu_i=0}^3 T^{\mu_1,\cdots,\mu_n} \sigma_{\mu_1}\otimes \cdots \otimes \sigma_{\mu_n}.
\end{equation}
This is a Hermitian operator with a trace of one, but it can have negative eigenvalues. The Hermitian operator with trace one plays a crucial role in investigating the temporal states, we will denote the set of such operators as $\Herm_1$. The collection of all PDOs will be denoted as $\PDO$.
The symmetric bloom construction of temporal state \cite{fullwood2023quantum,parzygnat2023timereversal} is also contained in $\Herm_1$.
Appendix~\ref{sec:causality} gives a more comprehensive discussion of PDO.

Classical space-time causality is a partial order relation $R(\Acal)\subset  \Acal\times \Acal$ over a collection of space-time events $\mathcal{A}=\{E(x,t)\}_{x,t}$, the causal relation between two events is determined by their corresponding space-time coordinates.
The PDO provides a framework to encode the quantum causality. Here, for an event set $\Acal$, we assign a Hilbert space $\mathcal{H}=\otimes_{e\in \Acal} \mathcal{H}_e$, a PDO over $\mathcal{H}$ can be regarded as a quantum generalization of the classical quantum relation $R(\Acal)$.
Since density operators can be regarded as spatial PDO, we see that the negativity of PDO is an indicator of the existence of the temporality of the event set $\Acal$.

\section{Pseudo-density operator marginal problem}
\label{sec:PDOMagProb}

In the conventional space-time causal marginal problem, we consider a family of sets of events $\mathfrak{M}_{\Acal} = \{\mathcal{A}_1, \cdots, \mathcal{A}_k\}$, referred to as a marginal scenario of $\mathcal{A} = \cup_i \mathcal{A}_i$. The problem is to determine if there exists a global causal structure $R(\mathcal{A})$ over all events in $\Acal$ that is consistent with the causal structures $R(\mathcal{A}_i)$ for all $i=1, \cdots, n$.
This problem is straightforward if all the causal structures $R(\Acal_i)$ are compatible with each other. In such cases, a solution always exists.

\begin{observation}
       The deterministic classical causal marginal problem always has a solution.
\end{observation}

It's worth pointing out that even for the classical probabilistic causal model, the marginal problem is highly non-trivial in general \cite{Allen2017quantum,barrett2019quantum} and largely unexplored.

Now let us consider the quantum case, where the causal structure is encoded by a PDO.
One of the most crucial features of PDOs is that the partial trace is well-defined. For a given set of events $\Acal$, if we make a bipartition $\Acal=\Acal_1\cup \Acal_2$, the reduced PDO can be defined as $R_{\Acal_1}=\Tr_{\Acal_2} R_{\Acal}$ (and similarly for $R_{\Acal_2}$).
Two PDOs $R_{\Acal}$ and $R_{\mathcal{B}}$ are called compatible if $\Tr_{\Acal\setminus \mathcal{B}}R_{\mathcal{A}}=\Tr_{\mathcal{B}\setminus \Acal} R_{\mathcal{B}}$, that is, they have the same reduced PDOs on their overlapping event set $\Acal \cap \mathcal{B}$.
The marginal scenario $\mathfrak{M}_{\Acal}$ for $\Acal$, in this case, consists of a collection of event sets $\Acal_1,\cdots,\Acal_n$ together with compatible PDOs $R_{\Acal_1},\cdots,R_{\Acal_n}$. We can define the following PDO marginal problem.

\begin{definition}[PDO marginal problem]
Consider a marginal scenario consisting of a family of event sets $\Acal_1,\cdots,\Acal_n$ with their corresponding PDOs $R_{\Acal_1},\cdots,R_{\Acal_n}$, such that they are compatible. The PDO marginal problem asks if there exists a global PDO $R_{\mathcal{A}}$ with $\Acal=\cup_i \Acal_i$ such that $R_{\Acal_i}=\operatorname{Tr}_{\Acal\setminus \Acal_i} R_{\Acal}$ for all $i=1,\cdots,n$.
\end{definition}

The PDO marginal problem always has a trivial solution if the marginal event sets do not overlap, $R_{\Acal}=\otimes_i R_{\Acal_i}$.
The problem becomes more complicated and interesting when the marginal event sets have non-empty overlapping.
When PDOs are not compatible, it's obvious that there is no solution to the PDO marginal problems.

From the previous discussion, we see that an $n$-event PDO is determined by a rank-$n$ tensor $T^{\mu_1,\cdots,\mu_n}$.  
Taking the partial trace over some event subset, we obtain the new tensor for the reduced PDO by just setting the corresponding indices as zero. 
For example, for $T^{\mu_1\mu_2\mu_3}$, tracing over the third event, the tensor of the reduced PDO is just $T^{\mu_1\mu_20}$.
This substantially simplifies the problem.

\begin{theorem}[$\mathbf{Herm}_1$ marginal problem]
\label{thm:HermMag}
Consider the marginal problem $\{R_{\Acal_i}\}_{i=1}^n$ with $R_{A_i}\in \PDO(\Acal_i)$ and  $\Acal=\cup_{i=1}^n\Acal_i$.
In $\Herm_1(\Acal)$, 
there always exists a solution $R$ which is the solution to the marginal problem. In other words, the marginal problem in $\Herm_1(\Acal)$ is trivial.
\end{theorem}

\begin{proof}
Before we give general proof, let's consider a simple example.
Suppose that $\Acal_1=\{1,2\}$, $\Acal_2=\{2,3\}$ and $\Acal_3=\{1,3\}$ (we use $1$ to denote $E_1$, etc.),
the corresponding qubit PDOs are $R_{\Acal_1}$ and $R_{\Acal_2}$ with their respective correlation tensor $T^{\mu_1\mu_2}_{\Acal_1}$, $T^{\mu_2\mu_3}_{\Acal_2}$ and $T^{\mu_1\mu_3}_{\Acal_3}$.
The compatibility condition over event $\Acal_1\cap \Acal_2$ is equivalent to $T^{0\mu_2}_{\Acal_1}=T^{\mu_2 0}_{\Acal_2}$ and similarly for others.
Our aim is to find a rank-3 tensor $T_{\Acal}^{\mu_1\mu_2\mu_3}$ such that $T^{\mu_1\mu_2}_{\Acal_1}$, $T^{\mu_2\mu_3}_{\Acal_2}$ and $T^{\mu_1\mu_3}_{\Acal_3}$ can be reproduced from it by setting the corresponding indices as zeros.
This can be solved by the following procedure:
(i) set $T_{\Acal}^{\mu_1\mu_2 0}=T^{\mu_1\mu_2}_{\Acal_1}$;
(ii) set $T_{\Acal}^{0\mu_2 \mu_3}=T^{\mu_2\mu_3}_{\Acal_2}$; 
(iii) set $T_{\Acal}^{\mu_1 0 \mu_3}=T^{\mu_1\mu_3}_{\Acal_3}$
(iv) set arbitrary real 
values to $T^{\mu_1 \mu_2 \mu_3}$ with $\mu_1,\mu_2,\mu_3\neq 0$.
It's clear that the solutions form a $3^3$ dimensional real vector space.
See Fig.~\ref{fig:Tensor} for an illustration.

In  this same spirit, we can prove the general statement using induction.
Suppose that for any $\{R_{\Acal_i}\}$ with $|\cup_i\Acal_i|\leq (n-1)$, there always exists a solution.
Now consider a set of PDOs with $|\cup_i\Acal|=n$, we divide the collection of event sets $\{\Acal_i\}$ into two classes: (i) those whose sizes are less than or equal to $n-2$, which we denote as $\Bcal_i$; (ii) those whose sizes are equal to $n-1$, which we denote as $\mathcal{C}_i$.
Notice that without loss of generalities, we assume that there is no $i,j$ such that $\Acal_i \subsetneq \Acal_j$.
The assumption for induction ensures that there is a marginal problem solution for the first class, $R_{\Bcal}$ with $\Bcal=\cup_i \Bcal_i$ and $|\Bcal|\leq n-1$.
We could consider the worst case that $|\Bcal|=n-1$.
The problem becomes a marginal problem for $\{\Bcal,\mathcal{C}_1,\cdots,\mathcal{C}_k\}$.
In the worst case, there $n$ such event sets, $\mathcal{C}_1=\{2,\cdots,n\}$,$\cdots$, $\mathcal{C}_{n-1}=\{1,\cdots,n-2,n\}$, $\Bcal={1,\cdots,n-1}$. We construct the correlation tensor $T^{\mu_1,\cdots,\mu_n}$ as follows: 
(i) set $T^{0\mu_2\cdots\mu_n}=T_{\mathcal{C}_1}^{\mu_2\cdots\mu_n}$, $T^{\mu_1 0 \mu_3\cdots \mu_n}=T_{\mathcal{C}_2}^{\mu_1\mu_3\cdots\mu_n}$, etc.; 
(ii) assign arbitrary real values to $T^{\mu_1\cdots \mu_n}$ with $\mu_1,\cdots,\mu_n \neq 0$.
This completes the proof. 
\end{proof}

\begin{figure}[t]
	\centering
	\includegraphics[width=4cm]{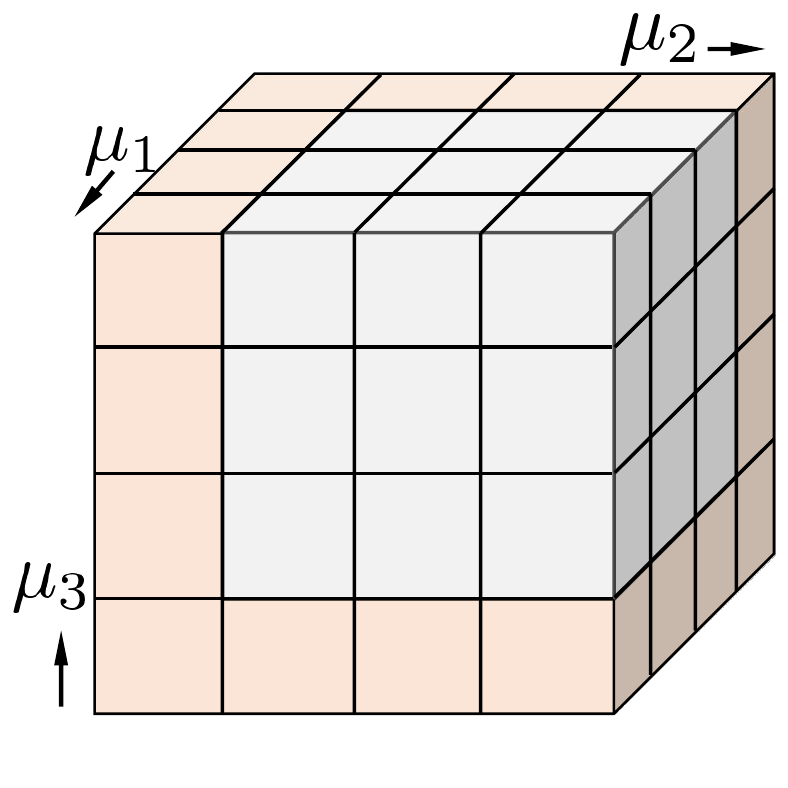}
	\caption{The illustration of the proof for $\Herm_1$ marginal problem, the cube represents the tensor $T^{\mu_1\mu_2\mu_3}$ of marginal problem solution $R_{\Acal}$. The light gray boxes represent the free parameter, while the light red boxes represent the parameters fixed by reduced PDOs $R_{\Acal_1},R_{\Acal_2},R_{\Acal_3}$.
	\label{fig:Tensor}}
\end{figure}

Notice that this theorem strongly depends on the existence of Hilbert-Schmidt operators, and this approach can also be applied to quantum state marginal problems.
We will denote the set of solutions for a given PDO marginal scenario $\mathfrak{M}_{\Acal}$ in $\Herm_1(\Acal)$ as $\mathbf{Marg}(\mathfrak{M}_{\Acal})$.
The solution for marginal problems in $\PDO(\Acal)$ is a subset of $\mathbf{Marg}(\mathfrak{M}_{\Acal})$.
In practice, there will be some other constraints to the solution. For example, in spatial cases, the pure state solution requires that the global state is a pure state; the bosonic solution requires the solution to be symmetric under permutation, and the fermionic solution requires the state to be antisymmetric under permutation.

When dealing with an event set containing a large number of events, symmetry is a useful tool.
We introduce the notion of quantum pseudo-channel (QPC) and symmetry for PDOs in Appendix~\ref{sec:QPC}, and we have the following result:
\begin{theorem}
    If the PDO marginal problem for a collection of PDOs $\mathcal{R}=\{R_{\Acal_1},\cdots,R_{\Acal_n}\}$ has a $G$-symmetric solution $R_{\Acal}$ with $\Acal=\cup_i\Acal_i$, then $G$ is also a symmetry of $\mathcal{R}$.
\end{theorem}

\begin{proof}
 Notice that 
 $\Tr_{\mathcal{A}\setminus \Acal_i^c}\Phi_g(R_{\Acal})=\Tr_{\mathcal{A}\setminus \Acal_i^c}(R_{\Acal})=R_{\Acal_i}$, the symmetry operation is just the marginal QPC of $\Phi_g$.
\end{proof}

Now, let's explore how these findings can be applied to various types of PDO marginal problems.

\subsection{Space-time separable marginal problem}

In Ref. \cite{navascues2021entanglement}, a special case of quantum state marginal problem is proposed, where they consider a collection of separable states and ask if there exists a global separable state and can reproduce all the given states as marginals.
We will call this a separable marginal problem (In Ref. \cite{navascues2021entanglement}, it's named as entanglement marginal problem).
In space-time state formalism, we can consider a similar problem.
But in this case, we need to introduce the notion of space-time separable states.
Consider an event set $\Acal$, we define the space-time product in the usual way $|a_1,\cdots,a_n\rangle =|a_1\rangle\otimes \cdots \otimes |a_n\rangle$, See Sec.~\ref{sec:SepRep} for details. Denote the set of all space-time product states as $\mathbf{Prod}(\Acal)$, and then the set of space-time separable states are just the convex hull $\mathbf{Sep}(\Acal)=\operatorname{Conv}(\mathbf{Prod}(\Acal))$.
The space-time separable state is thus of the form
\begin{equation}
W_{\Acal}=\sum_{a_1,\cdots,a_n}p(a_1,\cdots,a_n)\otimes_{i=1}^n |a_i\rangle \langle a_i|,
\end{equation}
where $p(a_1,\cdots,a_n)$ is a probability distribution. It's clear that
$W_{\Acal}$ is a positive semidefinite trace-one operator.

\begin{definition}[space-times separable marginal problem]
    For a marginal scenario $\mathfrak{M}_{\Acal}$ consisting of a given collection of event sets $\{\Acal_i\}$ with their corresponding  separable space-time separable states $\{W_{\Acal_i}\}$, the  space-times separable marginal problem asks if there exists a space-time separable state $W_{\Acal}$ for $\Acal=\cup_i\Acal_i$ such that all $W_{\Acal_i}$ can be reproduced by taking marginals.
\end{definition}

Let us now see how to use theorem~\ref{thm:HermMag} to solve this problem.
Combining the  theorem~\ref{thm:HermMag} and corollary~\ref{Coro:SepRep},
we know that there always exists a set of quasi-probabilistic separable solution
\begin{equation}
\begin{aligned}
     &\mathbf{Marg}(\mathfrak{M}_{\Acal})\\
     =&\{W_{\Acal}=\sum_{a_1,\cdots,a_n} p(a_1,\cdots,a_n)\otimes_{i=1}^n |a_i\rangle \langle a_i|\},
\end{aligned}
\end{equation}
where all $p(a_1,\cdots,a_n)$ are quasi-probability distributions.
To obtain the positive semidefinite solution set, we need first impose the positive semidefinite condition
\begin{equation}
\begin{aligned}
     &\mathbf{Marg}^{\rm pos}(\mathfrak{M}_{\Acal})\\
     =&\{W_{\Acal}\in \mathbf{Marg}(\mathfrak{M}_{\Acal})| \Tr (W_{\Acal}Y)\geq 0, \forall Y\geq 0\}.
\end{aligned}
\end{equation}
The second step is to choose the separable ones from these positive semidefinite solutions.
However, there is a more efficient approach to filter the solution from $\mathbf{Marg}(\mathfrak{M}_{\Acal})$ using the polytope approximation of $\mathbf{Sep}(\Acal)$, see Fig.~\ref{fig:Sep}.
Suppose that we have $n$ space-time separable states $R_1,\cdots,R_n$, they can generate a convex polytope $\mathbb{P}=\mathbf{Sep}(R_1,\cdots,R_n)=\operatorname{Conv}(R_1,\cdots,R_n)$.
By Minkowski-Weyl theorem, this polytope can be rewritten as a bounded intersection of half-spaces $\mathbb{P}=\cap_{i=1}^m\mathbb{H}_i$.
Each half-space is determined by an Hermitian operator $K_i$, namely $\mathbb{H}_i=\{R\in \Herm|\langle R,K_i\rangle \geq 0\}$.
The marginal problem solution contained in this polytope is thus
\begin{equation}
  \begin{aligned}
     &\mathbf{Marg}^{\mathbb{P}}(\mathfrak{M}_{\Acal}) \\
     =&\{W_{\Acal}\in \mathbf{Marg}(\mathfrak{M}_{\Acal})| \Tr (W_{\Acal}K_i)\geq 0, \forall i\}.
\end{aligned}  
\end{equation}
In this way, we obtain an operational method to solve the space-time separable marginal problem, which can be implemented numerically.
To minimize the computational complexity, we need to find an efficient way to determine the half-spaces from the extreme points of the polytope, this is the well-known \emph{convex hull problem} and it is proved to be an $\#$P-hard in general and  NP-hard for simplicial polytope \cite{dyer1983complexity}.

\begin{figure}[t]
	\centering
	\includegraphics[width=8cm]{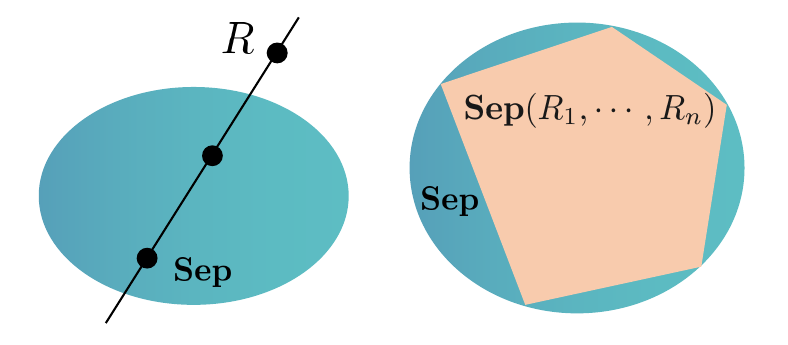}
	\caption{The left figure illustrates how to decompose a general PDO into a quasi-probabilistic mixture of space-times product state. For an Hermitian trace-one $R$, we can always find two separable states $W_1,W_2$ such that $R=\eta W_1+(1-\eta) W_2$ with $\eta \in \mathbb{R}$.
 The right figure illustrates the separable polytope constructed from a set of separable PDOs.
	\label{fig:Sep}}
\end{figure}

\subsection{Space-time symmetric extension}

Another crucial case of the marginal problem is the so-called symmetric extension \cite{Doherty2002distinguishing,Lami2019extendibility}, which has many applications in quantum information theory.
For the space-time state, we have a corresponding generalization.
Consider a two-event set $\Acal=\{A,B\}$ and its space-time state $W_{AB}$ as defined in definition~\ref{def:W} in Appendix~\ref{sec:causality}, the symmetric extension of $W$ is a $n$-event space-time state $W_{ABB_1\cdots B_{n-2}}$ such that all reduced space-time states satisfy $W_{AB_i}=W_{AB}$.
Here we show that in the space-time state framework, we always have a solution.

\begin{corollary}
For any two-event space-time state (for which  PDO is a special example) $W_{AB}$, the symmetric extension $W_{AB_1\cdots B_k}$ always exists in the space of all quasi-probabilistic mixture of space-time product state. 
\end{corollary}

\begin{proof}
 From corollary~\ref{Coro:SepRep} in Appendix, we see that $W_{AB}$ can be decomposed as 
 \begin{equation}
     W_{AB}=\sum_{ab}p(a,b) |a\rangle \langle a|\otimes |b\rangle \langle b|.
 \end{equation}
 The symmetric extension is given by
\begin{equation}
    W_{ABB_1\cdots B_{n-2}}=\sum_{ab}p(a,b) |a\rangle \langle a|\otimes (|b\rangle \langle b|)^{\otimes n-1}.
\end{equation}
 It's straightforward to verify that all reduced $W_{AB_i}=W_{AB}$.
\end{proof}

This technique can also be applied to extendibility for $m$-event $W_{A_1\cdots A_k B_1\cdots B_l}$ with respect to $B_1\cdots B_l$.
Notice the above corollary means that any $W\in \Herm_1$ is extendible in $\Herm_1$.
We prove it using the corollary~\ref{Coro:SepRep} in Appendix. This can also be transformed into a marginal problem and be proved using theorem~\ref{thm:HermMag}.
Suppose we have a collection of space-time state $W_{AB}=W_{AB_1}=\cdots W_{AB_{n-2}}$, theorem~\ref{thm:HermMag} ensures that there exists a non-empty solution set $\mathbf{Marg}(W_{AB},W_{AB_1},\cdots, W_{AB_{n-2}})$.
Then we can add more constraints to filter the solutions we need as we have done in the previous subsection.

\subsection{Polygamy of space-time correlations}

For spatial quantum correlations, it is well known that there are monogamy relations for entanglement, quantum steering, and Bell nonlocality. These monogamy relations impose restrictions on the distribution and sharing of these quantum correlations among multiple parties.
For example, in the case of a singlet state, it is not possible for three parties to share the state simultaneously. If Alice and Bob share the singlet state, then the state between Alice and Carol cannot be a singlet state.
The monogamy relation can be reformulated using a quantum marginal problem.
This means that the marginal scenario $\mathfrak{M}=\{\psi^-_{AB},\psi^-_{AC}\}$ has no solution.
However, for space-time correlations, the monogamy relation will be broken, an example has been given in Ref. \cite{marletto2019theoretical}.
Here, using the marginal problem framework, we see that polygamy is a general phenomenon for space-time states.

Let us take the singlet state $R_s$ in Eq.~\eqref{eq:singlet} as an example, to construct a symmetry extension $R_{AB_1\cdots B_n}$ with $R_{AB_i}=R_s$.
Using the method given in theorem \ref{thm:HermMag}, the correlation tensor of the marginal solution can be denoted as $T^{\nu \mu_1\cdots \mu_n}$.
The requirement of $\Tr_{B_2,\cdots,B_n}R_
{AB_1\cdots B_n}=R_{s}$ implies that $T^{\nu \mu_1 0\cdots 0}=T_s^{\nu\mu_1}$, etc.
This gives us 
\begin{equation}
    R_{AB_1\cdots B_n}=\frac{1}{d^{n+1}}(\mathds{I}^{\otimes n+1}-\sum_i \Omega_i+\Xi),
\end{equation}
where $\Omega_i=X_A\otimes \mathds{I}\otimes \cdots  \otimes \mathds{I} \otimes X_{B_i} \otimes \mathds{I} \otimes \cdots \otimes \mathds{I}+Y_A\otimes \mathds{I}\otimes \cdots  \otimes \mathds{I} \otimes Y_{B_i} \otimes \mathds{I} \otimes \cdots \otimes \mathds{I}+Z_A\otimes \mathds{I}\otimes \cdots  \otimes \mathds{I} \otimes Z_{B_i} \otimes \mathds{I} \otimes \cdots \otimes \mathds{I}$, and $\Xi$ is a free parameter term.

\subsection{Classical quasi-probability marginal problem}

The classical probability distribution marginal problem is crucial for us to understand Bell nonlocality and quantum contextuality in the non-signaling and more general no-disturbance framework.
Not all measurement statistics admit a joint probability distribution that can reproduce the measurement statistics as marginal distributions.
The vanishing of  a joint probability distribution is a criterion of quantumness exhibited in a quantum behavior.
To construct a joint probability distribution for nonlocal or contextual behavior, we must introduce negativity to the distribution.
This inspires us to consider the more general quasi-probability marginal problem since we have shown that the quasi-probability distribution arises naturally from space-time states.
Here we elaborate on how to relate the problem with a space-time state marginal problem.

Consider three quasi-random variables $a,b,c$ (namely $p(a),p(b),p(c)$ are quasi-probabilities), and quasi-probability distribution $p(a,b)$, $p(b,c)$,
no-disturbance means that $\sum_{a}p(a,b)=\sum_c p(b,c)$, we will also say the $p(a,b)$, $p(b,c)$ are compatible with each other in this situation.
With this definition of compatibility, we introduce the following definition of the classical marginal scenario.

\begin{definition}
Consider a set of quasi-random variables $\Acal=\{X_1,\cdots,X_n\}$, then a classical marginal scenario $\mathfrak{M}_{\Acal}$ on $\Acal$ is a non-empty collection  $\{\Acal_1,\cdots \Acal_k\}$ of subsets of $\Acal$ together with a set of compatible quasi-probability distributions $\{p(X\in\Acal_i)\}_{i=1}^k$.
\end{definition}

The quasi-probability marginal problem asks: if there exists a joint quasi-probability distribution $p(X\in \Acal)$ of all quasi-random variables in $\Acal$ such that all quasi-probability distributions in the  marginal scenario can be reproduced as marginal distributions of $p(X\in \Acal)$.
A marginal scenario can be represented by a graph $G[\mathfrak{M}_{\Acal}]$ (usually called compatibility graph), where each quasi-random variable is drawn as a vertex and all $\Acal_i$ are drawn as a hyper-edge (or equivalently, a clique, draw edges such that all vertices in $\Acal_i$ are connected with each other).
A well-known result for classical probability marginal problem states that, if $G[\mathcal{M}_{\Acal}]$ is a chordal graph, there exists a solution to the marginal problem. The proof can be found in, e.g. \cite{Ramanathan2012generalized}.
This result also holds for the quasi-probability marginal problem, the proof can be generalized straightforwardly.
For a classical marginal scenario $\mathfrak{M}_{\Acal}$, if its compatibility graph $G[\mathcal{M}_{\Acal}]$ is a chordal graph, there always exists a solution to the marginal problem.

Now let's see how to transform a quasi-probability marginal problem into a space-time marginal problem.
The main tool we will use is a generalization of the classical states \cite{SLLuo2008}, which we will call a space-time classical state.
Consider an $n$-event set $\Acal$, for each event $E_i$ we choose a complete set of rank-1 orthonormal projector $\{\Pi_{a_i}\}_{a_i}$, then we take the quasi-probabilistic mixture of the tensor product of these projectors
\begin{equation}
W_{\Acal}=\sum_{a_1,\cdots,a_n}p(a_1,\cdots,a_n)\Pi_{a_i}\otimes \cdots \Pi_{a_n}.
\end{equation}
Notice that in definition \ref{def:W} in Appendix, the local states may not be orthogonal with each other, thus the space-time state will exhibit quantum correlations.
For a quasi-probability distribution, all its classical space-time states are related by local unitary operations,
\begin{equation}
    W'_{\Acal}=(\prod_i U_i) W_{\Acal} (\prod_i U_i^{\dagger}).
\end{equation}
For a classical marginal scenario $\mathfrak{M}_{\Acal}$, there is a corresponding set of classical space-time states 
$\{W_{\Acal_i}\}_{i=1}^k$. We can define the following classical space-time state marginal problem:

\begin{definition}[classical space-time state marginal problem]
For a set of classical space-time states 
$\{W_{\Acal_i}\}_{i=1}^k$ which are compatible with each other up to local unitary operations, find a classical space-time state $W_{\Acal}$ such that all $W_{\Acal_i}$ are local unitary equivalent the reduced states of $W_{\Acal}$.
\end{definition}

\begin{theorem}
The quasi-probability classical marginal problem for a marginal behavior $\mathfrak{M}_{\Acal}$ is equivalent to the classical space-time state marginal problem $\{W_{\Acal_i}\}$.   
\end{theorem}

\begin{proof}
    We need to show that the existence of the quasi-probability marginal problem is equivalence to the existence of the classical space-time state marginal problem. 
    
    Suppose that $p(a_1,\cdots,a_n)$ is the solution of quasi-probability marginal problem for a marginal behavior $\mathfrak{M}_{\Acal}$, then we can choose arbitrary local complete rank-1 orthonormal projectors labeled as $\Pi_{a_i}$ and construct $W_{\Acal}=\sum_{a_1,\cdots,a_n} p(a_1,\cdots,a_n) \Pi_{a_1}\otimes \cdots \otimes \Pi_{a_n}$. It's easy to check that this is the solution of  classical space-time state marginal problem $\{W_{\Acal_i}\}$.

    For the other direction, suppose that $W_{\Acal}=\sum_{a_1,\cdots,a_n} q(a_1,\cdots,a_n)\Pi_{a_1}\otimes \cdots \otimes \Pi_{a_n}$ is a solution. The $W_{\Acal_i}$ is equal to $\Tr_{\Acal \setminus \Acal_i^c} W_{\Acal}$ up to local unitary operations implies that $\sum_{\Acal\setminus \Acal_i^c} q(a_1,\cdots,a_n) =p(a_{i_l}\in \Acal_i)$. This completes the proof.
\end{proof}

\section{Inferring global pseudo-density operator from reduced pseudo-density operator}
\label{sec:InfGlob}
Causal inference is the task of determining the causal structure underlying a set of random variables in the classical world \cite{pearl2009causality}, while its extension to the quantum realm is known as quantum causal inference. 
Recently, there has been a surge of interest in this field, with various approaches being explored based on different quantum causal models \cite{costa2016quantum,Giarmatzi2018a,bai2022quantum,cotler2019quantum,liu2023quantum}.
Our focus is on inferring the global causal structure from local causal structures. 
As we previously noted, the spectrum of a space-time state can be treated as a quasi-probability distribution.
Therefore, we introduce the concept of entropy for quasi-probability distributions and, by extension, for space-time states. 
Other attempts to introduce space-time entropy in alternative formalisms of space-time states and quantum stochastic processes are also underway \cite{cotler2018superdensity,lindblad1979non}, and exploring their interrelationships is a crucial topic for future research, which is left for our future study.
The space-time entropy of PDO provides a unified framework that encompasses both state entropy and dynamical entropy and has numerous potential applications. In this section, we present an information-theoretic approach for inferring the global space-time state from local reduced space-time states, using space-time entropy.

\subsection{Entropy of space-time states}

The concept of entropy indisputably plays a crucial role in modern physics.
Usually, the von Neumann entropy is defined on the spatial states $\varrho$, their spectra are regarded as a probability distribution $\Vec{\lambda}(\varrho)$. The von Neumann entropy $S(\varrho)$ is defined as the Shannon entropy $S(\Vec{\lambda}(\varrho))$ of $\Vec{\lambda}(\varrho)$.
Since space and time are treated equally in our framework, and the PDO $R$ represents our space-time state, we can naturally define the generalized von Neumann entropy as follows:
\begin{equation}
    S(R)=-\sum_i |\lambda_i| \log |\lambda_i|=-\Tr |R|\log |R|,
\end{equation}
where quasi-probabilities $\lambda_i$'s are eigenvalues of $R$ and $|R|=\sqrt{R^{\dagger}R}$.
When $R$ is a density operator, it becomes the von Neumann entropy.
The generalized R\'{e}nyi entropy is defined in a similar way
\begin{equation}
    S_{\alpha}(R)=\frac{1}{1-\alpha} \log \Tr |R|^{\alpha},
\end{equation}
and we have $\operatorname{lim}_{\alpha\to 1} S_{\alpha} (R)=S(R)$.

For two event sets $\Acal,\Bcal$, the conditional entropy, and mutual entropy are defined as follows
\begin{align}
    S(\Acal|\Bcal)=S({\Acal\Bcal})-S(\Bcal),\\
    I(\Acal:\Bcal)=S(\Acal)+S(\Bcal)-S(\Acal\Bcal),
\end{align}
where we denote $S(R_{\Acal})$ by $S(\Acal)$, and $R_{\Acal}$ is reduced PDO of $R_{\Acal \Bcal}$, etc.

The lack of positive semidefiniteness in space-time states results in the violation of several properties of entropy.
And this violation of properties is an indicator of the existence of temporal correlations. In this part, we will establish some properties of space-time entropy.

Recall that $\|R\|_{1}=\|\vec{\lambda}\|_1$ can be regarded as a causality monotone. More precisely, if we set $C(R)=(\|R\|_1-1)/2$, we see that \cite{fitzsimons2015quantum} 
\begin{enumerate}
    \item $C(R)\geq 0$ with $F(R)=0$ if $R$ is positive semidefinite (it also satisfies the normalization condition: $F(R)=1$ if $R$ is obtained from two consecutive measurements on a single qubit closed system).
    \item $C(R)$ is invariant under a local change of basis.
    \item $C(R)$  is non-increasing under local operations.
\end{enumerate}
We also have $\sum_i p_i C(R_i)\geq C(\sum_i p_i R_i)$. 
It's argued in \cite{Pisarczyk2019causal} that $F(R)=\log \| R\|_1$ is a causality monotone that is additive in the sense $F(R_1\otimes R_2)=F(R_1)+F(R_2)$, and $F(\sum_i p_iR_i)\leq \max_i \{F(R_i)\}$.
We now show that these two causality monotones appear naturally in the expression of the entropy of PDO.
From the spectrum quasi-probability distribution $\vec{\lambda}_R$ of $R$, we can construct a probability vector $\Vec{p}_R=(|\lambda_1|/\|R\|_1,\cdots,|\lambda_N|/\|R\|_1)$.
Then the Shannon entropy of $\Vec{p}_R$ is well-defined.
It's not difficult to check that 
\begin{equation} \label{eq:Entropy}
    S(\vec{\lambda}_R)=[2C(R)+1] [S(\Vec{p}_R)-F(R)].
\end{equation}
When there is no causality in $R$, $C(R)=F(R)=0$, we see $S(\vec{\lambda}_R)=S(\Vec{p}_R)$.
Thus, the equality above can also be used as a criterion for the existence of causality.

One of the primary purposes we introduce entropy of space-time states is to utilize the  generalized maximal entropy principle to infer the global space-time state from a set of reduced space-time states.
This is closely related to the space-time marginal problem.
However, to apply the maximal entropy principle, the entropy function should be 
upper-bounded. We now show that our definition of entropy satisfies this requirement.

\begin{theorem}
    For a given $n$-event set $\Acal$, the entropy is upper bounded, viz., there exists $K>0$ such that
    \begin{equation}
        S(R)\leq K, \quad \forall R\in \PDO(\Acal).
    \end{equation}
\end{theorem}

\begin{proof}
    Notice that $\Tr \sigma_{\mu}^2=d$, the sup-norm of $\sigma_{\mu}$ satisfies $\|\sigma_{\mu}\|_{\rm sup}\leq d^{1/2}$. Since $T^{\mu_1,\cdots,\mu_n}=\langle\{\sigma_{\mu_1},\cdots,\sigma_{\mu_n}\}\rangle$, $|T^{\mu_1,\cdots,\mu_n}|\leq d^{n/2}$.
    Then using triangle inequality of sup-norm, we see $\|R\|_1\leq d^n$ for all $R$.
    For any $d^n$-dimensional probability vector $\Vec{p}$, the Shannon entropy is upper bounded by $\log d^n=n\log d$.
    Then from Eq.~\eqref{eq:Entropy} we see that for all $R$, $S(R)$ is upper bounded by the same number.
\end{proof}

For the convenience of our later discussion, we also introduce the space-time relative entropy between $R_1,R_2\in \PDO(\Acal)$:
\begin{equation}
    S(R_1||R_2)=\Tr (|R_1|\log |R_1|) -\Tr (|R_1|\log |R_2|),
\end{equation}
where $|R_i|=\sqrt{R_i^{\dagger}R_i}$ for $i=1,2$.
Recall that Klein's inequality claims that, for the convex real function $f$ with derivative $f'$ and Hermitian operators $A,B$, we have
\begin{equation}
    \Tr[f(A)-f(B)-(A-B)f'(B)]\geq 0.
\end{equation}
Take $f(x)=x\log x$ we obtain
\begin{equation} \label{eq:Klein}
    \Tr A \log A-\Tr A \log B \geq \Tr (A-B).
\end{equation}
Set $A=|R_1|$, $B=|R_2|$, we obtain the lower bound of space-time entropy
\begin{equation}\label{eq:QrelEntropy}
    S(R_1||R_2)\geq \Tr (|R_1|-|R_2|)=2(C(R_1)-C(R_2)).
\end{equation}
When $R_1$ and $R_2$ are spatial states ($C(R_1)=C(R_2)=0$), it implies the non-negativity of the quantum relative entropy of the density matrix. 
And for PDOs such that the amount of causality of $R_1$ is greater than or equal to that of $R_2$, the relative entropy is also non-negative.

Recall that Lieb's concavity theorem claims that for any matrix $X$ and $0\leq t \leq 1$, the function 
\begin{equation}
    f(A,B)=\Tr (X^{\dagger}A^tXB^{1-t})
\end{equation}
is jointly concave in positive semidefinite $A$ and $B$.
Set $G_t(A,X)=\Tr (X^{\dagger} A^t X A^{1-t}) -\Tr (X^{\dagger} X A)$, Lieb's theorem implies that $G_t(A,X)$ is concave in positive semidefinite $A$, and this further implies that $G'_0(A,X)=\frac{d}{dt} G_t(A,X)|_{t=0}=\Tr (X^{\dagger} (\log A) X A)-\Tr(X^{\dagger}X(\log A)A)$ is concave in positive semidefinite $A$.
Then set
\begin{equation}
    \begin{aligned}
        A=\left(
\begin{array}{cc}
    |R_1| & 0 \\
     0 &  |R_2|
\end{array}
        \right)
    \end{aligned}, \quad  \begin{aligned}
        X=\left(
\begin{array}{cc}
    0 & 0 \\
    \mathds{I} &  0
\end{array}
        \right)
    \end{aligned},
\end{equation}
we see that $G'_0(A,X)=-S(R_1||R_2)$.
From this, we see that space-time relative entropy is joint convex in $|R_1|$ and $|R_2|$.

\begin{theorem}
    The entropy $S(R)$ of space-time state $R\in \PDO(\Acal)$ satisfy the following properties:
\begin{enumerate}
    \item Unitary invariant: $S(URU^{\dagger})=S(R)$ where $U$ is unitary operator.
    \item Weak additivity:  $S(R_1\otimes R_2)=[2C(R_2)+1] S(R_1)+[2C(R_1)+1]  S(R_2)$.
    \item Weak concavity: $\alpha S(R_1)+ (1-\alpha) S(R_2) \leq S(\alpha |R_1|+(1-\alpha) |R_2|)$.
\end{enumerate}  
\end{theorem}
    
\begin{proof}
    The proofs of 1 and 2 are straightforward.

    3. Set $A=|R_1|$ and $B=\alpha |R_1|+(1-\alpha)|R_2|$ in Eq.~\eqref{eq:Klein}, we obtain 
    \begin{equation}
        \begin{aligned}
            &\Tr |R_1|\log |R_1|-\Tr |R_1| \log \alpha |R_1|+(1-\alpha)|R_2|\\
    \geq &(1-\alpha) \Tr (|R_1|-|R_2|).
        \end{aligned}
    \end{equation}
    Similarly, set $A=|R_2|$ and $B=\alpha |R_1|+(1-\alpha)|R_2|$ in Eq.~\eqref{eq:Klein}, we obtain 
        \begin{equation}
        \begin{aligned}
            &\Tr |R_2|\log |R_2|-\Tr |R_2| \log \alpha |R_1|+(1-\alpha)|R_2|\\
     \geq & \alpha \Tr (|R_2|-|R_1|).
        \end{aligned}
    \end{equation}
   Multiplying the first inequality by $\alpha$ and the second by $(1- \alpha)$ and adding them yields the conclusion.
\end{proof}

\begin{theorem}[weak subadditivity]\label{thm:wsubadd}
For a PDO $R_{\Acal\Bcal}$, let $\Delta(R_{\Acal\Bcal})=\Tr[ (|R_{\Acal\Bcal}|-|R_{\Acal}|\otimes |R_{\Bcal}|)\log (|R_{\Acal}|\otimes |R_{\Bcal}|)]$  we have the following weak subadditivity relation
\begin{equation}\label{eq:subAdd}
    \begin{aligned}
        & S(R_{\Acal})+S(R_{\Bcal})-S(R_{\Acal \Bcal}) \\
        \geq & \Delta(R_{\Acal\Bcal})+ \Tr (|R_{\Acal\Bcal}|-|R_{\Acal}|\otimes |R_{\Bcal}|).
    \end{aligned}
\end{equation}
Notice that, when reduced PDOs $R_{\Acal}$ and $R_{\Bcal}$ are positive semidefinite, we see that $\Delta(R_{\Acal\Bcal})=0$; when $R_{\Acal,\Bcal}$ is positive semidefinite, the right-hand side of the above expression becomes zero. 
\end{theorem}

\begin{proof}
        This can be proved by taking $R_1=R_{\Acal\Bcal}$ and $R_2=R_{\Acal}\otimes R_{\Bcal}$ in Eq.~\eqref{eq:QrelEntropy}.
\end{proof}

These entropic inequalities also impose some constraints on the space-time marginal problems.
When restricted to the spatial density operators, they give entropic constraints for the classical \cite{fritz2012entropic} and  quantum state marginal \cite{osborne2008entropic,carlen2013extension}.
Its extension in the quantum channel marginal problem is also briefly discussed in \cite{haapasalo2021quantum}.

Let's take temporal two-event qubit PDO as an example to calculate the entropy. Set $\rho(t_1)=(\mathds{I}+\vec{r}\cdot \vec{\sigma})/2$ with $\vec{r}=(r_1,r_2,r_3)$ and set $\mathcal{E}^{t_1\to t_2}=\operatorname{id}$. From Eq.~\eqref{eq:PDOQubit}, we obtain the spectrum quasi-probability distribution
\begin{equation}
    \vec{\lambda}_R=(-\frac{1}{2},\frac{1}{2},\frac{1}{2}(1-\|\vec{r}\|), \frac{1}{2}(1+\|\vec{r}\|)).
\end{equation}
We see that entropy satisfies $0\leq S(R)\leq 2$ (see Fig. \ref{fig:EntropyQubit}).

\begin{figure}[t]
    \centering
     \includegraphics[scale=0.60]{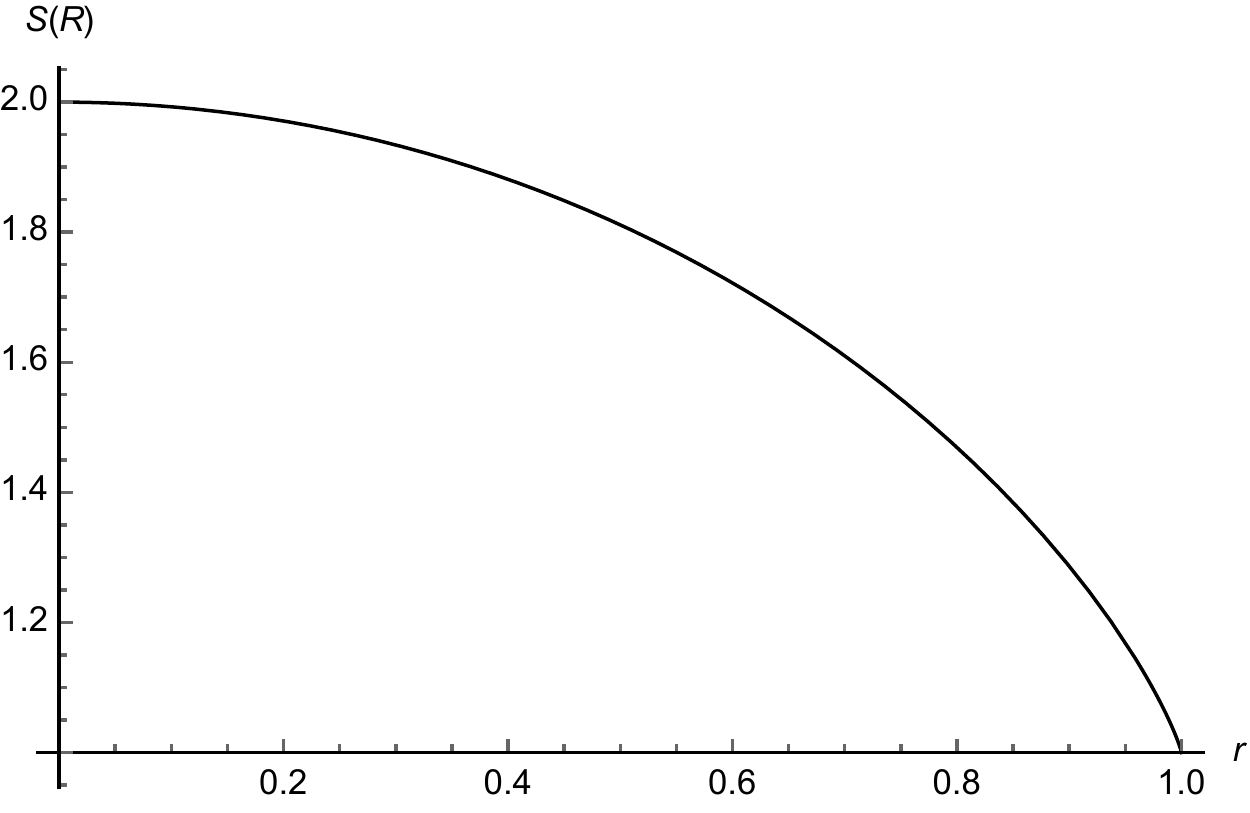}
    \caption{The space-time entropy of two-event PDO for two consecutive measurements over a qubit state, where $r$ is the norm of the Bloch vector.}
    \label{fig:EntropyQubit}
\end{figure}

We would like to stress that in the space-time state formalism, the information on the dynamical process of spatial states is also contained in the space-time states. Thus the entropy of space-time states can also be used to investigate the dynamical entropy.

\subsection{Inferring global pseudo-density operator via maximum entropy principle}

It's also possible to extend various existing generalizations of the concept of entropy of density operators to space-time states.
Regardless of a particular generalization of the concept of entropy, the key usage of entropy crucially hinges on the maximum entropy principle and its various deductions \cite{Jaynes1957,Jaynes1957a}.
Here we propose the following space-time maximum entropy principle:

\begin{principle}[Space-time maximum entropy principle]
    For a given set of constraints $\{L_k(R)=0\}$ of space-time state $R$, the best inference of space-time state is the one that maximizes the entropy $S(R)$ subject to these constraints.
    More precisely, using the Lagrange multiplier method, it's the one that maximizes the functional ${L}(R)=S(R)-\sum_k \alpha_k L_k(R)$.
\end{principle}

Let's now apply the space-time maximum entropy principle to the marginal problem, viz., given a collection of reduced space-time states, we try to infer the best global space-time state that can reproduce these reduced states as marginals.
Suppose that we have the information of a set of marginal space-time states $\mathcal{M}=\{R_{\Acal_k}\}$, then there is a set of constraints
\begin{equation}
  L_k(R_{\Acal}) = \Tr_{\Acal\setminus\Acal_k} R_{\Acal}-R_{\Acal_k}=0.
\end{equation}
Utilizing the maximum entropy principle we can infer the global space-time state as
\begin{equation}
    R_{\Acal}^{\mathcal{M}}=\operatorname{argmax}_{R_{\Acal}}\{{L}(R_{\Acal})=S(R_{\Acal})-\sum_k \alpha_k L_k(R_{\Acal}) \}.
\end{equation}
This optimization problem can be solved using different methods, in the Appendix, we will introduce the neural network approach.
Notice that, in the space-time scenario, the space-time state we obtain is usually not unique.
For example, consider two single-event states $R_{A}=R_{B}=\mathds{I}/2$, the entropies of both the spatial states $R_s=R_A
\otimes R_A$ and the PDO in Eq.~\eqref{eq:PDOQubit} by setting $\rho(t_1)=\mathds{I}/2$ and $\mathcal{E}^{t_1\to t_2}=\operatorname{id}$ reach maximum value $2$.
This reflects the fact that non-overlapping marginal space-time states are not enough to determine the global space-time state.

Notice that $R_{\Acal}^{\mathcal{M}}$ is our best inference of the global space-time state with the local information $\mathcal{M}=\{R_{\Acal_k}\}$ in hand.
One interesting application of this fact is that the space-time correlation exhibit in $R_{\Acal}$ that beyond the information contained in $\mathcal{M}$ can be characterized by comparing the original $R_{\Acal}$ with our inference $R_{\Acal}^{\mathcal{M}}$.
Consider a space-time state $R_{\Acal}$, let's denote all its $k$-event reduced space-time states as $\mathcal{M}_k$.
Using the maximum entropy principle, we obtain the corresponding inference $R_{\Acal}^{\mathcal{M}_k}$.
Then for a given norm of operators $\|\cdot \|$, we define $C_k=\| R_{\Acal} -R_{\Acal}^{\mathcal{M}_k} \|$.
If $C_k>0$, $R_{\Acal}$ exhibits the genuine $(k+1)$-event space-time correlations, namely, the space-time correlation of $R_{\Acal}$ cannot be recovered with only $k$-reduced space-time state information.

\section{Quantum pseudo-channel marginal problem}

The dynamics of a PDO are characterized by a quantum pseudo-channel (QPC), see Appendix~\ref{sec:QPC} for a detailed discussion. In this section, we will consider the marginal problem of the QPC and demonstrate that it can be transformed into a PDO marginal problem through channel-state duality.


The notion of marginal marginal quantum channel is introduced in \cite{Hsieh2022quantum}.
This can be naturally generalized to the QPC.
Suppose that $\Acal$ and $\Bcal$ are input and out event sets of the QPC $\Phi_{\Bcal|\Acal}$.
The marginal is defined with respect to a bipartition of both the input and output event sets.
Let $\mathcal{X}\subset \Acal$ and $\mathcal{Y}\subset \Bcal$, the marginal QPC $\Phi_{\mathcal{Y}|\mathcal{X}}$ is defined as follows:
for arbitary $R_{\Acal}\in \PDO(\Acal)$ we have 
\begin{equation}\label{eq:magC}
    \Tr_{\mathcal{Y}^c} \Phi_{\Bcal|\Acal}(R_{\Acal})= \Phi_{\mathcal{Y}|\mathcal{X}} (\Tr_{\mathcal{X}^c} (R_{\Acal})),
\end{equation}
where $\mathcal{X}^c$ and $\mathcal{Y}^c$ are complements of $\mathcal{X}$ and $\mathcal{Y}$ in $\Acal$ and $\Bcal$. 
We will denote this marginal QPC as $\Tr_{\mathcal{Y}^c|\mathcal{X}^c} \Phi_{\Bcal|\Acal} =\Phi_{\mathcal{Y}|\mathcal{X}}$.

Hereinafter, for convenience of discussion, we will use a normalized Choi-Jamio{\l}kowski representation of $\Phi_{\Bcal|\Acal}$, 
\begin{equation}
    J(\Phi_{\Bcal|\Acal})=\frac{1}{d_{\Acal}} \Phi_{\Bcal|\Acal}(E_{ij})\otimes E_{ij}.
\end{equation}
It's clear that $\Phi_{\Bcal|\Acal}(R)/d_{\Acal}=\Tr_{\Acal} [J(\Phi_{\Bcal|\Acal}) (\mathds{I} \otimes R^T)]$.
We will call this correspondence channel-state duality.
Using the channel state duality, we can translate this defining condition \eqref{eq:magC} into a state form (see, e.g., \cite[Appendix A]{Hsieh2022quantum} and references therein)
\begin{equation}
    \Tr_{\Ycal^c} J(\Phi_{\Bcal|\Acal})=J(\Phi_{\Ycal|\Xcal}) \otimes \frac{\mathds{I}_{\Xcal^c}}{d_{\Xcal^c}}.
\end{equation}
Since we take a different convention for the Choi-Jamio{\l}kowski map, there is no dimension factor here in our expression.
This implies that the Choi map for the marginal channel is indeed the marginal state $J(\Phi_{\Ycal|\Xcal})=\Tr_{\Ycal^c|\Xcal^c} J(\Phi_{\Bcal|\Acal})$.

Let us now show that the QPC marginal problem can be transformed into a space-time state marginal problem.
Then we can invoke the results for PDO marginal problem to investigate the QPC marginal problem.

\begin{definition}[QPC marginal problem]
    Given a collection of QPC $\{\Phi_{\Bcal_i|\Acal_i}\}$, suppose that they are compatible with each other, the QPC marginal problem asks if there exists a global QPC from event set $\Acal=\cup_i\Acal_i$ to $\Bcal=\cup_i\Bcal_i$ which can reproduce all QPCs by taking marginals. 
\end{definition}

From channel-state duality, $J(\Phi_{\Bcal|\Acal})$ is Hermitian if and only if $\Phi_{\Bcal|\Acal}$ is HP.
$\Phi_{\Bcal|\Acal}$ is TP implies that
$\Tr_{\Bcal}J(\Phi_{\Bcal|\Acal})=\mathds{I}_{\Acal}/d_{\Acal}$, thus $\Tr J(\Phi_{\Bcal|\Acal})=1$. 
When $\Phi_{\Bcal|\Acal}$ is HPTP, $J(\Phi_{\Bcal|\Acal})\in \Herm_1$.
From the previous discussion, we see that the compatibility of two QPCs on their overlap is indeed the same as the compatibility of states corresponding to them.
The observations and findings discussed above can be summarized as the following result:

\begin{theorem}[$\HPTP$ marginal problem]
    For a collection of compatible QPC $\{\Phi_{\Bcal_i|\Acal_i}\}$, there always exists a solution for the marginal problem in $\HPTP(\Acal,\Bcal)$. 
\end{theorem}

\begin{proof}
    Theorem \ref{thm:HermMag} guarantees that there exists a  $(m+n)$-rank tensor $T^{\nu_1\cdots\nu_m\mu_1\cdots \mu_n}$ such that 
    \begin{equation}
J_{\Bcal|\Acal}=\sum_{\mu_i,\nu_j}T^{\nu_1\cdots\nu_m\mu_1\cdots \mu_n} \sigma_{B_1}^{\nu_1}\otimes \cdots \otimes \sigma_{B_m}^{\nu_m} \otimes \sigma_{A_1}^{\mu_1}\otimes \cdots \otimes \sigma_{A_1}^{\mu_n}
    \end{equation}
    is a solution of the $\Herm_1$ state marginal problem $\{J(\Phi_{\Bcal_i|\Acal_i})\}$.
    We only need to show that there exist one $J_{\Bcal|\Acal}$ such that $\Tr_{\Bcal}J_{\Bcal|\Acal}=\mathds{I}_{d_{\Acal}}/d_{\Acal}$.
    This is clear from the fact that when $\nu_1=\cdots=\nu_m=0$, $T^{0\cdots0\mu_1\cdots \mu_n}\neq 0$ only if $\mu_1,\cdots,\mu_n=0$. 
    Since $\Tr J_{\Bcal|\Acal}=1$, $T^{0,\cdots,0}=1/d_{\Bcal}d_{\Acal}$, we arrive at the conclusion.
\end{proof}

\section{Conclusion and discussion}

In this work, we discussed the marginal problem for space-time states and space-time channels.
We show that for space-time states, the solution to the marginal problem almost always exists.
We discuss several applications of this result, including space-time separable marginal problem, space-time symmetric extension, and polygamy of space-time correlations, classical quasi-probability marginal problem.
Via the channel-state duality, we show that the space-time channel marginal problem can be reformulated as a space-time state marginal problem. Thus the result of the space-time marginal problem can be directly applied to the space-time channel marginal problems.
We also introduce an approach to inferring the global space-time state from a given set of reduced space-time states based on the generalized maximum entropy principle.

Despite the significant progress made in understanding space-time states and quantum causality models, several open problems still need to be addressed. One such issue is finding the physical realization of an arbitrary given PDO, which remains a challenge for most proposals of space-time states and quantum causality models. Another critical problem is explaining the polygamy of space-time correlations, which is closely related to the former problem. Although we have shown that almost all space-time marginal problems have solutions in the PDO framework, it is essential to understand the physics behind this phenomenon. One suggestion based on the open time-like curve circuit is given in \cite{marletto2019theoretical}, and we plan to investigate this further in our future studies.

In addition, while we have defined space-time entropy in our framework of space-time states and discussed its properties, a deeper understanding and investigation of the difference and connection between existing proposals for space-time entropy and dynamical entropy is still needed \cite{cotler2018superdensity,lindblad1979non}. Another interesting problem is the application of the maximum entropy principle in quantum causal inference. Although the existing quantum causal inference protocol is mainly based on process matrix \cite{costa2016quantum,Giarmatzi2018a,bai2022quantum}, a possible definition of causal entropy (or space-time entropy) is provided in Ref. \cite{cotler2018superdensity}. Using the maximum causal entropy principle to infer the global causal structure from local causal structures in the process matrix formalism is also of great interest. All of these issues will be left for our future studies.

\acknowledgements{
We acknowledge Fabio Costa, Xiangjing Liu, James Fullwood, and Arthur Parzygnat for bringing our attention to some related works.  Special thanks are extended to Arthur Parzygnat for insightful discussions regarding theorem~\ref{thm:wsubadd}.
Z. J. and D. K. are supported by National Research Foundation in Singapore and A*STAR under its CQT Bridging Grant. M. S. is supported by the National Research Foundation, Singapore, and Agency for Science, Technology and Research (A*STAR) under its QEP2.0 programme (NRF2021-QEP2-02-P06), the Singapore Ministry of Education Tier 1 Grants RG77/22, the Singapore Ministry of Education Tier 2 Grant MOE-T2EP50221-0005, grant no.~FQXi-RFP-1809 (The Role of Quantum Effects in Simplifying Quantum Agents) from the Foundational Questions Institute and Fetzer Franklin Fund (a donor-advised fund of Silicon Valley Community Foundation) 
}


\appendix

\section{Quantum space-time causality and pseudo-density operator formalism}
\label{sec:causality}

\subsection{Pseudo-density operator}

The original PDO is introduced for the qubit system \cite{fitzsimons2015quantum}, when dealing with the higher dimensional system, one needs to embed the system into the space of the many-qubit system and restrict the evolution to the appropriate subspace.
However, here we will take a different approach, we assume that the local space is of arbitrary $d$ dimensions and the measurements are generalized Pauli operators (a.k.a., Hilbert-Schmidt operators) $\sigma_{\mu}$, $\mu=0,\cdots,d^2-1$ which are Hermitian operators satisfying (i) $\sigma_0=\mathds{I}$;  (2) $\Tr( \sigma_{j})=0$ for all $j \geq 1$; (3) These matrices are orthogonal $	\Tr(\sigma_{\mu}\sigma_{\nu})=d\delta{\mu \nu}$.
They form a basis for the real vector space of $d\times d$ Hermitian operators. An explicit example is generalized Gell-Mann matrices (GGM) \cite{Gell-mann1962symmetry} (See \cite[Sec. 2]{wei2022antilinear} for an explicit matrix expression we will use).
When $d=2$, they become Pauli operators. The continuous variable version of PDO is introduced in Ref. \cite{zhang2020different}.
In this work, we only consider the finite-dimensional case.

The pseudo-density operator formalism concerns the following scenario: we have a quantum system distribution over space and we choose to measure some (generalized) Pauli measurements over some qudit ($x$) at some particular instant in time ($t$). 
We introduce a tensor product structure among all space-time events $\mathcal{A}=\{E(x_i,t_i)\}_{i=1}^n$. Thus the total space is $\mathcal{H}_{\Acal}=\otimes_{i}\mathcal{H}[E(x_i,t_i)]$.
In this way, we obtain a state of the system that is distributed over space-time
\begin{equation} \label{eq:PDOnd}
    R_{\mathcal{A}}=\frac{1}{d^n} \sum_{\mu_1,\cdots,\mu_n=0}^{d^2-1}
   T^{\mu_1\cdots \mu_n} \otimes_{j=1}^n \sigma_{\mu_j},
\end{equation}
where $T^{\mu_1\cdots \mu_n}=\langle 
    \{ \sigma_{\mu_j}\}_{j=1}^n\rangle$ is the expectation value of a collection of Pauli measurements.
This $R_{\Acal}$ is called a PDO.
Notice that when all qudits are measured at the same instant of time, we obtain the normal Bloch representation of a multipartite state \cite{wei2022antilinear}. 
We will denote the set of all PDOs for an event set $\Acal$ as $\PDO(\Acal)$.

It's useful to introduce a quantum circuit representation of the causal structure behind the PDO.
See Fig. \ref{fig:PDO} for an illustration.
The input state is a (possibly multipartite) state $\varrho(t_0)$ and we will always denote the time instant for the input state as $t_0$.
Suppose that there are $n$ instants in time that we are concerned with, $t_1,\cdots,t_n$. During every two consecutive instants $t_i$ and $t_{i+1}$, we can apply some quantum operations $\mathcal{E}^{t_i\to t_{i+1}}$ over the state.
The space coordinates are represented by the quantum wire $x_1,\cdots,x_m$, and the event $E(x_i,t_j)$ is  just  measuring (generalized) Pauli operators of $x_i$-state at time instant $t_j$.
For a collection $\Acal=\{E(x,t)\}$ of space-time events, we will obtain a corresponding pseudo-density operator $R_{\Acal}$.
It's crucial that $\mathcal{E}^{t_i\to t_{i+1}}$ has a given structure that describes the propagation of causality over the time interval $[t_i,t_{i+1}]$.
For example of the causal structure as in Fig.~\ref{fig:PDO}, $\mathcal{E}^{t_1\to t_{2}}=\mathcal{E}^{t_1\to t_{2}}_{x_1}\otimes \mathcal{E}^{t_1\to t_{2}}_{x_2x_3x_4}\otimes \mathcal{E}^{t_1\to t_{2}}_{x_5x_6x_7}$, the effect of event $E(x_3,t_1)$ is propagated to event $E(x_3,t_2)$, but it's not propagated to event $E(x_6,t_2)$ due to the existence of tensor product structure of $\mathcal{E}^{t_1\to t_{2}}$.
Actually, during two consecutive time instants, there may exist a complex quantum circuit that characterized the propagation of causality, which is also under extensive investigation \cite{mi2021information}. 
This quantum circuit representation of the PDO is convenient to investigate the transformation of PDOs, which will be rigorously defined and studied later.
A fixed background causal structure has a fixed quantum circuit. The event set is embedded into the space-time structure determined by the circuit.
We will denote the set of PDOs obtained by embeded $\Acal$ into a circuit $\mathsf{C}$ as $\PDO(\Acal,\mathsf{C}[\varrho(t_0),\{\mathcal{E}^{t_i\to t_{i+1}}\}])$.
The probabilistic mixture of PDOs is also allowed, thus $\PDO(\Acal,\mathsf{C}[\varrho(t_0),\{\mathcal{E}^{t_i\to t_{i+1}}\}])$ can be regarded as the convex hull of the PDO obtained from the given circuit.

It turns out that a complete characterization of the set of PDOs for a given set of events is a very complicated problem.
Only for the single-qubit two-event case, the spatial and temporal PDO sets are fully characterized \cite{bengtsson2017geometry,zhao2018geometry}.  
From the definition of a PDO $R$,  we see that it must satisfy \cite{fitzsimons2015quantum}: (i) $R$ is Hermitian; (ii) $R$ is trace-one.
Another natural requirement that PDO must satisfy is that all single-event reduced PDO must be positive semidefinite \cite{horsman2017can}.
For $n$-event set $\Acal=\{E_i\}_{i=1}^n$, each event has its associated Hilbert space $\mathcal{H}_{E_i}$, and the Hilbert space of the whole event set is then given by the tensor product of each Hilbert spaces, i.e $\mathcal{H}_{\Acal} = \otimes_i \mathcal{H}_{E_i}$.
It's convenient to introduce the set of all trace-one Hermitian operators
    \begin{equation}
        \mathbf{Herm}_{1} ({\Acal})=\{R\in \mathbf{B}(\mathcal{H}_{\Acal})|R^{\dagger}=R, \Tr(R)=1\},
    \end{equation}
where $\mathbf{B}(\mathcal{H}_{\Acal})$ denotes the set of all bounded operators over $\mathcal{H}_{\Acal}$, $\mathbf{Herm}$ denotes the set of all Hermitian operators and the subscript denotes the trace of these operators.
It's clear that $\PDO(\Acal)\subset  \mathbf{Herm}_{1} ({\Acal})$.
A subtle thing is that the correlation function should be bounded for fixed settings of measurement choice. 
And for spatial correlations, the positive semidefinite condition needs to be imposed.
Another interesting and closely relevant open question is, for a given PDO, how to find a quantum process to realize it. This also goes beyond the scope of this paper, and we leave it for our future study.

\begin{example}[Two-event PDO]
\label{exp:TwoPDO}
The simplest PDO is the one obtained by measuring two-point correlation functions $\langle \sigma_{\mu_1}(x_1,t_1)\otimes \sigma_{\mu_2}(x_2,t_2)\rangle$ over the (possibly multipartite) qubit state.
There are three causally distinct situations: 
\begin{enumerate}
    \item The spatial two-qubit PDO, this corresponds to the case $t_1=t_2=t$ and $x_1\neq x_2$, 
    \begin{equation}
        \begin{aligned}
		\begin{tikzpicture}
                \draw[line width=.6pt,black] (0.4,-0.4) -- (0.4,0);	
			\draw[line width=.6pt,black] (-0.4,-0.4) -- (-0.4,0);
                \draw[line width=.6pt,black] (0.4,0.5) -- (0.4,0.9);	
			\draw[line width=.6pt,black] (-0.4,0.5) -- (-0.4,0.9);
                \draw[line width=0.6 pt, fill=gray, fill opacity=0.2] 
		(-0.6,0) -- (.6,0) -- (.6,0.5) -- (-.6,0.5) -- cycle; 
			\node[ line width=0.6pt, dashed, draw opacity=0.5] (a) at (0.6,-0.8){$(x_2,t)$};
			\node[ line width=0.6pt, dashed, draw opacity=0.5] (a) at (-0.4,-0.4){$\bullet$};
                \node[ line width=0.6pt, dashed, draw opacity=0.5] (a) at (0.4,-0.4){$\bullet$};
			\node[ line width=0.6pt, dashed, draw opacity=0.5] (a) at (-0.6,-0.8){$(x_1,t)$};
    		\node[ line width=0.6pt, dashed, draw opacity=0.5] (a) at (0,0.25)
      {$\mathcal{E}_{x_1x_2}$};
		\end{tikzpicture}
	\end{aligned}
    \end{equation}
    In this case, $R_{(x_1,t),(x_2,t)}=\varrho_{x_1x_2}(t)$ is positive semidefinite trace-one operator.
    
    \item The temporal two-qubit PDO, this corresponds to the case $x_1=x_2=x$ and $t_1\neq t_2$,
        \begin{equation}
        \begin{aligned}
		\begin{tikzpicture}
                \draw[line width=.6pt,black] (0.4,-0.4) -- (0.4,0);	
			\draw[line width=.6pt,black] (-0.4,-0.4) -- (-0.4,0);
              \draw[line width=.6pt,black]   (0.4,0.5) -- (0.4,0.9);	
			\draw[line width=.6pt,black] (-0.4,0.5) -- (-0.4,0.9);
                \draw[line width=0.6 pt, fill=gray, fill opacity=0.2] 
		(-0.6,0) -- (.6,0) -- (.6,0.5) -- (-.6,0.5) -- cycle; 
			\node[ line width=0.6pt, dashed, draw opacity=0.5] (a) at (-0.4,-0.4){$\bullet$};
                \node[ line width=0.6pt, dashed, draw opacity=0.5] (a) at (-0.4,0.9){$\bullet$};
			\node[ line width=0.6pt, dashed, draw opacity=0.5] (a) at (-0.4,-0.8){$(x,t_1)$};
   	      \node[ line width=0.6pt, dashed, draw opacity=0.5] (a) at (-0.4,1.3){$(x,t_2)$};
    		\node[ line width=0.6pt, dashed, draw opacity=0.5] (a) at (0,0.25)
      {$\mathcal{E}_{x_1x_2}$};
		\end{tikzpicture}
	\end{aligned}
    \end{equation}
In this case,  using the Stinespring extension, we can just consider the general quantum channel $\mathcal{E}_x^{t_1\to t_2}$ acting on $\varrho_{x}(t_1)$. The corresponding PDO is of the form \cite{zhao2018geometry}
\begin{equation}\label{eq:PDOQubit}
    R_{(x,t_1),(x,t_2)}=(\operatorname{id}\otimes \mathcal{E}_x^{t_1\to t_2})(\{\varrho_x(t_1)\otimes \frac{\mathds{I}}{2},\mathsf{SWAP}\}),
\end{equation}
where we have used the anti-commutator bracket and $\mathsf{SWAP}=\sum_{\mu=0}^3\sigma_{\mu}\otimes \sigma_{\mu}/2$.
Another equivalent expression based on Jordan's product of state and Choi matrix of the channel is given in \cite{horsman2017can}.

    \item The hybrid space-time PDO, this corresponds to the case $x_1\neq x_2$ and $t_1\neq t_2$, $R_{(x_1,t_1),(x_2,t_2)}$.
        \begin{equation}
        \begin{aligned}
		\begin{tikzpicture}
                \draw[line width=.6pt,black] (0.4,-0.4) -- (0.4,0);	
			\draw[line width=.6pt,black] (-0.4,-0.4) -- (-0.4,0);
              \draw[line width=.6pt,black] (0.4,0.5) -- (0.4,0.9);	
			\draw[line width=.6pt,black] (-0.4,0.5) -- (-0.4,0.9);
                \draw[line width=0.6 pt, fill=gray, fill opacity=0.2] 
		(-0.6,0) -- (.6,0) -- (.6,0.5) -- (-.6,0.5) -- cycle; 
			\node[ line width=0.6pt, dashed, draw opacity=0.5] (a) at (0.6,1.3){$(x_2,t_2)$};
			\node[ line width=0.6pt, dashed, draw opacity=0.5] (a) at (-0.4,-0.4){$\bullet$};
                \node[ line width=0.6pt, dashed, draw opacity=0.5] (a) at (0.4,0.9){$\bullet$};
			\node[ line width=0.6pt, dashed, draw opacity=0.5] (a) at (-0.6,-0.8){$(x_1,t_1)$};
    		\node[ line width=0.6pt, dashed, draw opacity=0.5] (a) at (0,0.25)
      {$\mathcal{E}_{x_1x_2}$};
		\end{tikzpicture}
	\end{aligned}
    \end{equation}
    See Ref. \cite{liu2023quantum} for a general expression of PDO for a given quantum circuit.
\end{enumerate}
\end{example}

The most spatially correlated two-event PDOs are the well-known Bell states \cite{Horodecki2009}, e.g. singlet state $\psi^-$, 
\begin{equation}\label{eq:singlet}
   R_{s}=|\psi^-\rangle \langle \psi^-|=\frac{1}{4}(\one \otimes \one-X\otimes X-Y\otimes Y -Z\otimes Z).
\end{equation}
The strongest temporally correlated two-event PDOs are arguably the ones obtained from by measuring a given state for two consecutive time instants, like
\begin{equation}\label{eq:temporalBell}
    R_t=\frac{1}{4}(\one \otimes \one+X\otimes X+Y\otimes Y +Z\otimes Z),
\end{equation}
which has been used to implement quantum teleportation in time \cite{marletto2021temporal}.
When taking the partial trace over one of two events for both $R_s$ and $R_t$, we will obtain the single qubit maximally mixed state.
In the spatial case, $R_s$ is known as the maximally entangled state, thus $R_t$ can be regarded as a maximally entangled temporal state in a similar spirit.

A negative eigenvalue of $R$ signifies that the causal structure is not purely spatial, that is, some temporal causal mechanisms are embodied.
This implies that a big difference between temporal and spatial PDO is that spatial PDO can be pure (rank of $R$ could be one), but temporally correlated PDO can not be a pure state.

\subsection{Quasi-probabilistic mixture of  space-time product states}
\label{sec:SepRep}
From the definition of a PDO $R$, we know that its eigenvalues $\vec{\lambda}(R)$ are real (possibly negative) numbers such that $\sum_i \lambda_i=1$, , that is, $R$ can be written as 
\begin{equation}\label{eq:Reigen}
    R=\sum_i \lambda_i |\psi_i\rangle \langle \psi_i|,
\end{equation}
with eigenvectors $\psi_i$ corresponding to $\lambda_i$.
This means that the spectrum of a PDO can be regarded as a quasi-probability distribution, which has been investigated from aspects since Wigner's pioneering work \cite{Wigner1932on} and turns out to play a crucial role in quantum foundations \cite{Spekkens2008negativity,ferrie2011quasi}, quantum optics \cite{scully1987quasi}, quantum computation \cite{Delfosse2015wigner}, etc.
Utilizing the information-theoretic tools developed in quasi-probability distribution to investigate the properties of PDOs is also interesting, which will be done later.

Due to the possibility of the existence of negative eigenvalues, it appears that the eigenvalue of PDO $R$ may not be bounded.
However, this is indeed not the case. 
Notice that $\Tr \sigma_{\mu}^2=d$ for all $\mu$, the sup-norm satisfy $\|\sigma_{\mu}\|_{\rm sup}\leq \sqrt{d}$.
This further implies that $\|\sigma_{\mu_1}\otimes \cdots \otimes \sigma_{\mu_n}\|_{\rm sup} \leq d^{n/2}$.
Since $T^{\mu_1,\cdots,\mu_n}$'s are correlation functions of $\{\sigma_{\mu_1},\cdots,\sigma_{\mu_n}\}$, we also have $|T^{\mu_1\cdots \mu_n}|\leq d^{n/2}$.
Then, we  see
\begin{equation}
    \begin{aligned}
        \|R_{\Acal}\|_{\rm sup}&\leq \frac{1}{d^n} \sum |T^{\mu_1\cdots \mu_n}|\|\sigma_{\mu_1}\otimes \cdots \otimes \sigma_{\mu_n}\|_{\sup}
        &\leq d^n.
    \end{aligned}
\end{equation}

We would like to stress that the physical interpretation of the spectrum for a general temporally correlated PDO is still lacking.
Here we will propose one possible interpretation based on the following theorem.

\begin{theorem}[Quasi-probability separable expansion]
	\label{thm:QuasiRep}
    Consider two-event set $\Acal$, any PDO $R_{\Acal}\in \PDO(\Acal)$ can be represented by a quasi-probabilistic mixture of product states. Namely, there exists a quasi-probability distribution $P(a,b)$ and states $|a,b\rangle=|a\rangle \otimes |b\rangle$ such that 
    \begin{equation}
        R_{\Acal}=\sum_{a,b} P(a,b) |a,b\rangle\langle a,b|.
    \end{equation}
\end{theorem}

\begin{proof}
 Recall that any bipartite pure state $|\psi\rangle$ can be written as 
 \begin{equation}
     |\psi\rangle \langle \psi|=\sum_{a,b} \eta(a,b) |a,b\rangle \langle a,b|,
 \end{equation}
 where $\eta(a,b)$ is a quasi-probability distribution and $|a,b\rangle =|a\rangle \otimes |b\rangle$.
See Fig.~\ref{fig:Sep} for an illustration (also see \cite{Vidal1999robustness}).
 From Eq.~\eqref{eq:Reigen}, for $\lambda_i$ is a quasi-probability distribution, and each $\psi_i$ gives a corresponding quasi-probability distribution $\eta_{i}(a_i,b_i)$. 
 The function $P(a_i,b_i)=\lambda_i \eta_{i}(a_i,b_i)$ is a quasi-probability distribution, thus we obtain a quasi-probability separable expansion of $R_{\Acal}$.
\end{proof}

\begin{corollary}
\label{Coro:SepRep}
	For any $n$-event set $\Acal$, any PDO $R_{\Acal}$ can be expressed as a quasi-probabilistic mixture of  pure space-time product states
\begin{equation}\label{eq:RWdecom}
R_{\Acal}=\sum_{a_1,\cdots,a_n}p(a_1,\cdots,a_n)|a_1,\cdots,a_n\rangle \langle a_1,\cdots, a_n|,
\end{equation}
	where $|a_1,\cdots,a_n\rangle =|a_1\rangle \otimes \cdots \otimes |a_n\rangle$.
\end{corollary}

\begin{proof}
This can be proved via repeatedly taking bipartition of the event set and using the theorem~\ref{thm:QuasiRep}.
More precisely, if we take the bipartition of the event set as $A|B$, then from Eq.~\eqref{eq:Reigen} and theorem  \ref{thm:QuasiRep}, $|\psi_i\rangle \langle \psi_i|=\sum_{a,b}\eta(a,b)|\phi_a\rangle \langle \phi_a|\otimes |\xi_b\rangle\langle\xi_b|$. We then take bipartition of $A=A_1|A_2$ and $B=B_1|B_2$, $|\phi_a\rangle \langle \phi_a|$ and $|\xi_b\rangle\langle\xi_b|$ can decompose into quasi-probability mixtures. By repeating this procedure, we will obtain the required expression.
\end{proof}

\begin{figure}[b]
	\centering
	\includegraphics[width=7.5cm]{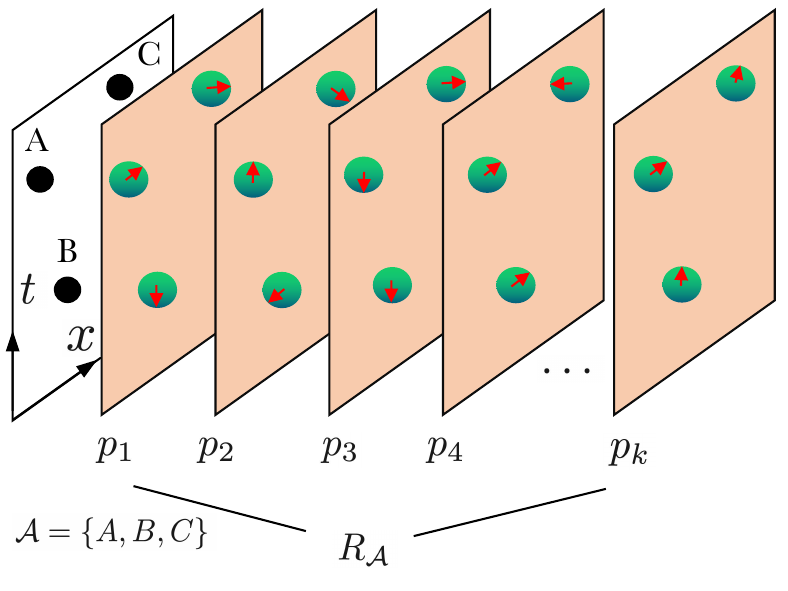}
	\caption{The illustration of the quasi-probabilistic mixture interpretation of PDO. The event set $\Acal=\{A,B,C\}$, and $\vec{p}=(p_1,\cdots,p_k)$ is a quasi-probability vector. The PDO is a quasi-probabilistic mixture of product space-time states $R_{\Acal}=\sum_{i=1}^k p_i|a_i,b_i,c_i\rangle \langle a_i,b_i,c_i|$.
	\label{fig:PDORep}}
\end{figure}

Actually, the above results hold for arbitrary trace-one Hermitian operators, since in the proof we only use the trace-one condition and Hermiticity of $R$.
It's clear that this quasi-probabilistic decomposition is not unique. The concept of hidden negativity, which refers to the minimum negativity of the quasi-probability distribution utilized for representing a PDO as depicted in Eq.~\eqref{eq:RWdecom}, will be elaborated on extensively in our forthcoming work \cite{jia2023space}.
Inspired by the above results, we can introduce a more general formalism for space-time correlations.

\begin{definition}[Quasi-probabilistic mixture representation of space-time correlation]
\label{def:W}
Consider an $n$-event space-time scenario $\Acal=\{E_1,\cdots,E_n\}$, we still assign a local Hilbert space $\mathcal{H}_{E_i}$ for each event $E_i$. The local state vectors are independent, viz., they are in product-form $|a_1,\cdots,a_n\rangle=|a_1\rangle \otimes \cdots\otimes|a_n\rangle$. The correlations are captured by the negativity of quasi-probability distribution $\vec{p}=(p_1,\cdots,p_n)$,
\begin{equation} \label{eq:WA}
	W_{\Acal}=\sum_{i=1}^kp(a_1,\cdots,a_n) |a_1,\cdots,a_n\rangle \langle a_1,\cdots,a_n|.
\end{equation}
See Fig.~\ref{fig:PDORep} for an illustration.
\end{definition}

This quasi-probabilistic mixture representation of space-time correlation is of their own interest and we will discuss it in detail elsewhere \cite{jia2023space}.
The above result shows that the PDO formalism can be subsumed into this more general formalism.
When $\vec{p}$ is a probability vector, there is no quantum space-times correlation in $W_{\Acal}$. However, when there exist negative probabilities, there must be quantum space-time correlations.
Hereinafter, in the most general setting, we will call a matrix $W$ a space-time state if: (i) $W$ is Hermitian; (ii) $\Tr W=1$; and (iii) for any fixed event set, $\|W\|_{\rm sup}$ is upper bounded.
Any space-time state can be expressed as in Eq.~\eqref{eq:WA}.

\subsection{Space-time purification}
\label{sec:prurification}

For a PDO $R_{\Acal}$, due to the existence of negativity, it's impossible to purify in the usual way.
Nevertheless, we can still remedy this issue by introducing a more general form of purification, which is named space-time purification.

For a PDO $R_{\Acal}$, we have polar decomposition $R_{\Acal}=U_{\Acal}|R_{\Acal}|$, where $|R_{\Acal}|=\sqrt{R_{\Acal}^{\dagger}R_{\Acal}}$.
Then we can purify $|R_{\Acal}|$ via 
\begin{equation}
    |\Psi_{\Acal\Bcal}\rangle=\sum_{i}\sqrt{|\lambda_i|} |\psi_i\rangle \otimes |e_i\rangle,
\end{equation}
where $|e_i\rangle$'s are the orthonormal basis for the Hilbert space of an auxiliary system $\Bcal$.
The PDO $R_{\Acal}$ can be expressed as 
\begin{equation}
    R_{\Acal}=U_{\Acal} \Tr_{\Bcal}  |\Psi_{\Acal\Bcal}\rangle\langle \Psi_{\Acal\Bcal}|. 
\end{equation}
The main difference between the space-time purification with that of the mixed density operator is $\| \Psi_{\Acal\Bcal}\|\geq 1$. 
If $\| \Psi_{\Acal\Bcal}\|> 1$, there must be temporal correlations in $R_{\Acal}$.

\section{Quantum pseudo-channel}
\label{sec:QPC}
As we have seen, the PDO codifies the space-time correlations of a given event set. It's natural to consider the transformation among these PDOs, this naturally leads to the concept of quantum pseudo-channel (QPC). QPC can thus be regarded as space-time channels, this has not been discussed before. 
The only work we are aware of is \cite{cotler2018superdensity}, where the concept of space-time channel is briefly discussed in the superdensity operator formalism.
Since PDO formalism is completely different from that of superdensity operator, in superdensity operator formalism, the state is still positive semidefinite.
It's thus worth to discussing the definition and representation of QPCs in reasonable detail.

\subsection{Quantum pseudo-channel as higher-order maps}

In a straightforward way, we define QPC as a linear superoperator that maps pseudo-density operators to pseudo-density operators. 
All quantum channel is a special case of QPC, where the input and output state are both spatial density operators.

\begin{definition}[QPC]
Consider the space of all bounded operators over the Hilbert space $\mathcal{H}_{\Acal_X}=(\mathbb{C}^d)^{\otimes n_X}$ with $X=I, O$ (`in' and `out'), a pseudo-density channel is a linear map $\Phi:\mathbf{B}(\mathcal{H}_{\Acal_I}) \to \mathbf{B}(\mathcal{H}_{\Acal_O})$ such that $\Phi(R_{\Acal_I})\in \mathbf{PDO}({\Acal_O})$ for all $R_{\Acal_I}\in \mathbf{PDO}({\Acal_I})$, \emph{viz}.,
it maps PDO to PDO. We denote the corresponding set of QPC as $\QPC(\Acal_I,\Acal_O)$.
\end{definition}

The above definition of QPC can naturally be generalized to space-time states, which we will call space-time channels.
From the definition of a QPC $\Phi$, we see that $\Phi$ must satisfy: (i) it's Hermiticity-preserving (HP); (ii) it's trace-preserving (TP). There should also be some other constraints, e.g. the boundedness condition for PDO, and every physically realizable PDO must be mapped to a physically realizable PDO. However, the characterization of the set of physically realizable PDOs is still an open problem. At this stage, we will ignore these subtle issues and focus on general properties that QPCs must satisfy. The set of all HPTP maps will be denoted as $\mathbf{HPTP}(\Acal_I,\Acal_O)$. It's clear that $\QPC(\Acal_I,\Acal_O)\subset \mathbf{HPTP}(\Acal_I,\Acal_O)$.

We now introduce several different representations of QPC that will be useful for our later discussion.
Consider a superoperator $\Phi:\mathbf{B}(\mathcal{H}_{\Acal_I}) \to \mathbf{B}(\mathcal{H}_{\Acal_O})$, we have the following representations
\begin{enumerate}
\item The natural representation $N(\Phi)$. Using the vector map $||i\rangle \langle j| \rrangle=|i\rangle |j\rangle$, we define $N(\Phi):|R\rrangle\mapsto |\Phi(R)\rrangle$.
    
\item The Choi-Jamio{\l}kowski representation $J(\Phi)$. Let $E_{ij}=|i\rangle \langle j|$, $J(\Phi) 
	= \sum_{i,j} \Phi(E_{ij}) \otimes E_{ij}$.

 \item Kraus operator-sum representation $\Phi(R)=\sum_a A_a R B_a^{\dagger}$, where $A_a,B_a\in \mathbf{B}(\mathcal{H}_{\Acal_I},\mathcal{H}_{\Acal_)})$ for all $a$.

\item Stinespring representations $\Phi(R)=\Tr_{\mathcal{X}}(A RB^{\dagger})$, where $A,B\in \mathbf{B}(\mathcal{H}_{\Acal_I},\mathcal{H}_{\Acal_O} \otimes \mathcal{X})$, and $\mathcal{X}$ is an auxiliary space.
\end{enumerate}
In each of the above representations, there have  been well-established theories of HP and TP, see, e.g. \cite{hill1973linear,dePillis1967linear,choi1975completely}.
In Kraus operator-sum representation, an HPTP map is of the form
\begin{equation}
    \Phi(R)=\sum_a \lambda_a A_a R A_a^{\dagger},
\end{equation}
where $\lambda_a$ are real (possibly negative) numbers and $A_a$ satisfy $\sum_a \lambda_a A_a^{\dagger}A_a=\mathds{I}$.
The quantum channels (CPTP maps) are special cases of the QPC, for which we must have $\lambda_a \geq 0$.
Using the relation between Kraus representation and Choi-Jamio{\l}kowski representation, we obtain 
\begin{equation}
    J(\Phi)=\sum_a \lambda_a |A_{a}\rrangle \llangle A_a|,
\end{equation}
since $\lambda_a$ is in general not non-negative, we see that $J(\Phi)$ is not positive semidefinite but only Hermitian. And from TP condition, we have $\Tr_{\Acal_O} J(\Phi)=\mathds{I}$.
The Stinespring representation could also be obtained by setting $A=\sum_a \lambda_a A_a\otimes e_a$ and $B=\sum_a A_a\otimes e_a$ with $e_a$ an auxiliary orthonormal basis. The TP condition results in $A^{\dagger}B=\mathds{I}$.

Many properties of spatial quantum channels can be generalized to QPC.
These properties are crucial for us to understand the space-time correlations in a unified framework and may also have potential applications in quantum information processing in both space and time settings. Here we give an example of no-cloning theorem of space-time states:
There is no QPC that can perfectly clone an arbitrary given PDO. 
Suppose that there is a QPC $\Phi$ such that for all $R\in \PDO$ we have $\Phi(R)=R\otimes R$.
Consider two PDOs $R_1,R_2$ and their probabilistic mixture $R=pR_1+(1-p)R_2$, acting $\Phi$ on both sides, we will obtain a contradiction.
This is a direct result of the linearity of QPC.

Notice that this means that not just spatially distributed density operators cannot be cloned arbitrarily, but neither the temporally distributed state cannot be cloned arbitrarily.

The above definition of QPC is general but difficult to handle.
Let's give an example via the quantum circuit representation of PDOs.

In a most naive way, the classical deterministic causal structure for a given set of events $\mathcal{A}=\{E_1,\cdots, E_n\}$ is determined by the spacetime coordinates of these events. 
For two events $E_i,E_j$, depending on their spacetime coordinates, there is a corresponding causal relation between them.
If $E_j$ is in the light-cone of $E_i$,  there is a partial order: (i) $E_j\preceq E_i$ when $E_j$ in the past of $E_i$; (ii) $E_j\succeq E_i$ when  $E_j$ in the future of $E_i$. Otherwise, there is no order relation between them.
This equipped the event set $\mathcal{A}$ with a partial order relation $R(\mathcal{A})\subseteq \mathcal{A}\times \mathcal{A}$, which satisfy: $E_i\preceq E_i$; $E_i\preceq E_j$ and $E_j\preceq E_i$ implies $E_i=E_j$; $E_i\preceq E_j$ and $E_j\preceq E_k$ implies $E_i\preceq E_k$.
The causal relation $R(\mathcal{A})$ can be represented by a directed graph with each event represented by a vertex and each causal relation pair represented by a directed edge.
In abstract language, $\mathcal{A}$ is a vertex set and $R(\mathcal{A})$ is the edge set.

Consider two event sets $\Acal$ and $\Bcal$ with their respective causal relations $R(\Acal)$ and $R(\Bcal)$, a cause-effect preserving map $f:\Acal \to \Bcal$ is the one that preserves the causal order, i.e., if $E_i\preceq E_j$, then $f(E_i)\preceq f(E_j)$.
A cause-effect preserving QPC attached to a classical cause-effect preserving map $f:\Acal\to \Bcal$ is defined as follows.
We embed $\Acal$ and $\Bcal$ into two quantum circuits, then we assign a QPC that maps $R_{\Acal}$ to $R_{\Bcal}$.
Consider a circuit realization of a PDO with initial state $\rho(t_0)$, the quantum operations $\{\mathcal{E}^{t_i\to t_{i+1}}\}$. The QPC can be realized as a higher-order map in this situation, namely, a collection of maps of quantum operations $\Phi^{t_i\to t_{i+1}}(\mathcal{E}^{t_i\to t_{i+1}}) ={\mathcal{E}'}^{t_i\to t_{i+1}}$.

\subsection{Space-time Lindbladian and symmetry}

The previous discussion of QPC mainly focused on the transformation perspective of PDO. We could also treat these QPC as a dynamic process of PDOs, this leads to the conception of Lindbladian (or quantum Liouvillian) for a PDO.
Suppose that the event set $\Acal$ is controlled by some parameter $\tau$, then the corresponding PDO $R_{\Acal}(\tau)$ also depends on this parameter.
The dynamics of the PDO thus can be written as 
\begin{equation}
	\frac{d}{d\tau}R_{\Acal}(\tau)=\mathcal{L}(R_{\Acal}(\tau)).
\end{equation}
The detailed derivation of the above equation will be omitted here, it's in a spirit similar to the one for a spatially correlated system.
The space-time steady state  $R_{\Acal}(\infty)$ is defined as the solution of equation $\frac{d}{d\tau}R_{\Acal}(\tau)=0$, which is equivalent to $\mathcal{L} (R_{\Acal}(\infty))=0$.

\begin{definition}[Symmetries of PDO]
Consider a collection of PDOs $\mathcal{R}=\{R_{\Acal_1},\cdots,R_{\Acal_n}\}$, a $G$-symmetry of $\mathcal{R}$ is a group $G$ equipped with a representation for each $i$, $g\mapsto \Phi_g^{i}\in\mathbf{QPC}$ such that $\Phi_g^i(R_{\Acal_i})=R_{\Acal_i}$ for all $g\in G$.
\end{definition}

\begin{remark}
    In \cite{wei2022antilinear} the antilinear quantum channels are investigated, which are crucial for describing the discrete symmetries of an open quantum system and characterizing the quantum entanglement of the mixed quantum state. 
    For the pseudo-density operator, we can also introduce the antilinear quantum pseudo-channel.
\end{remark}

\section{The neural network approach to inferring the global space-time state}

In this part, let's introduce the neural network representation of space-time states and explain how to use it to solve marginal problems.
The neural network representation of quantum many-body states and density operators is a very powerful tool in solving various physical problems \cite{jia2019quantum,Carleo2019machine}.
Since a PDO is a generalization of a density operator, it's natural for us to consider the neural network representation of PDOs.
Unlike the density operator case, where we use a neural network to represent the matrix entries or coefficients of the purified states, here we will use a neural network to describe the correlation function.
This may be of independent interest for solving open-system problems.

\begin{figure}[t]
	\centering
	\includegraphics[width=6cm]{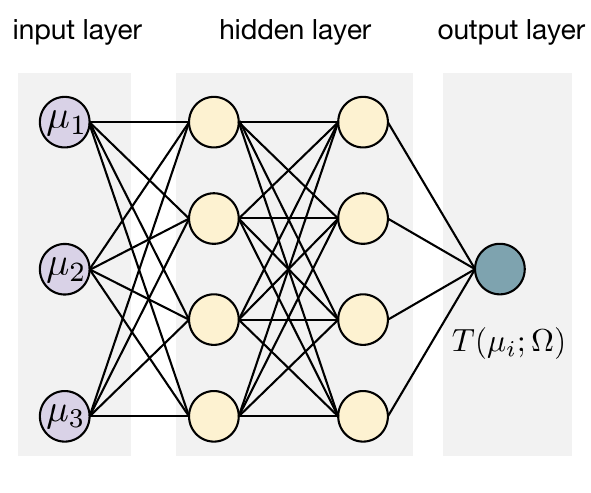}
	\caption{Illustration of a feedforward neural network representation of PDO.
	\label{fig:NNPDO}}
\end{figure}

Consider the PDO given in Eq.~\ref{eq:PDOnd}, we regard $T^{\mu_1\cdots \mu_n}$ as an $n$-variable function $T(\mu_1,\cdots,\mu_n)$.
The Hermicity is encoded in the realness of this function, and the trace-one condition is encoded in the $T^{0,\cdots 0}=1$. 
In order to simplify the discussion, hereinafter we will focus on the qubit PDO.
To represent $T^{\mu_1\cdots \mu_n}$, we build a neural network with $n$ visible neurons, where each visible neural represents $\mu_j$.
The neural network parameters, like connection weights, and biases are denoted as $\Omega=\{w_{ij},b_j\}$. 
For each given value of neural network parameters, we obtain corresponding PDO with correlation function given by $T_{NN}(\mu_1,\cdots,\mu_n;\Omega)$.
To ensure that $T^{0\cdots 0}=1$, we can normalize the function with $T^{\mu_1\cdots \mu_n}=T_{NN}(\mu_1,\cdots,\mu_n;\Omega)/T_{NN}(0,\cdots,0;\Omega)$.
To make it clearer, let's take a feedforward neural network  as an example (see Fig.~\ref{fig:NNPDO}).
Each neuron has several inputs $x_i$ with the corresponding weights $w_i$, there is bias $b$ and an activation function $f$ associated to the neuron, thus the output is 
\begin{equation}
    y=f(\sum_i w_ix_i-b).
\end{equation}
Using this basic building block, we can build a network, which consists of three different layers: input layer, hidden layer, and output layer, as shown in Fig.~\ref{fig:NNPDO}.
There are many different activation functions to be chosen from, here for qubit PDO, we can simply choose a function whose range is $[-1,1]$.
A frequently used one is $\operatorname{tanh}(x)=\frac{e^{x}-e^{-x}}{e^{x}+e^{-x}}$.
For each given set of weights and biases, the neural network outputs a function $T(\mu_j;\Omega)$. Then we use this output to write down a PDO $R(\Omega)$ which depends on the neural network parameters.
Thus the neural network PDO can be regarded as a variational space-time state.

To apply the neural network representation of PDO to the marginal problem, we need to maximize the Lagrangian functional ${L}(R(\Omega))$ (or equivalently, minimize $-{L}(R(\Omega))$) over neural network parameters $\Omega$.
This can be solved by the gradient descent method.
In this way, powerful machine-learning techniques can be applied to solve the problem of space-time correlations, not only for the marginal problem but also for many other problems, like determining the $k$-genuine space-time correlations, solving the steady state for a given Lindbladian, etc.

It's also worth mentioning that we use feedforward neural network states to build PDO, many other neural networks can also be used for representing PDO, like convolutional neural networks, Boltzmann machine, and so on.
The physical properties are encoded in the neural network structures of the representation.
The applications of the neural network approach in this direction are largely unexplored, this will be left for our future study.


%

\end{document}